\documentclass[11pt]{article}
\usepackage{amsfonts}
\usepackage{amsthm}
\usepackage{fullpage}
\usepackage{algorithmic}
\usepackage{algorithm}
\usepackage{graphicx}
\usepackage{subfigure}
\usepackage{epsfig}
\newcommand{\OUTPUT}{\STATE \textbf{output~}}
\newcommand{\RET}{\STATE \textbf{return~}}

\newcommand{\rank}{\mathtt{rank}}
\newcommand{\access}{\mathtt{access}}
\newcommand{\select}{\mathtt{select}}
\newcommand{\rc}{\mathtt{count}}
\newcommand{\rr}{\mathtt{report}}
\newcommand{\libcds}{{\sc Libcds}}
\newcommand{\B}{\tilde{B}}
\newcommand{\no}[1]{}

\newtheorem{lem}{Lemma}

\usepackage{multirow}

\title{Efficient Compressed Wavelet Trees over Large Alphabets
\thanks{An early partial versions of this article appeared in 
{\em Proc. SPIRE 2012} \cite{CN12}.}}

\date{}

\author{%
\begin{tabular}{ccc}
  Francisco Claude 
& Gonzalo Navarro\thanks{Funded in part by Millennium Nucleus Information and 
	Coordination in Networks ICM/FIC P10-024F, Chile.}
& Alberto Ord\'o\~nez\thanks{Xunta de Galicia (co-funded with FEDER), 
	ref. 2010/17 and CN 2012/211 (strategic group), and by 
	the Spanish MICINN ref. AP2010-6038 (FPU Program)} \\
  Esc. Inf. \& Tel. 
& Dept. of Computer Science 
& Database Laboratory \\
  Univ. Diego Portales, Chile
& Univ. of Chile, Chile
& Univ. da Coru\~na, Spain \\
  {\tt fclaude@recoded.cl} 
& {\tt gnavarro@dcc.uchile.cl} 
& {\tt alberto.ordonez@udc.es}
\end{tabular}}

\begin{document}
\maketitle

\begin{abstract}
The {\em wavelet tree} is a flexible data structure that permits representing
sequences $S[1,n]$ of symbols over an alphabet of size $\sigma$, within 
compressed space and supporting a wide range of operations on $S$. When
$\sigma$ is significant compared to $n$, current wavelet tree representations
incur in noticeable space or time overheads. In this article we introduce
the {\em wavelet matrix}, an alternative representation for large alphabets
that retains all the properties of wavelet trees but is significantly 
faster. We also show how
the wavelet matrix can be compressed up to the zero-order entropy of the 
sequence without sacrificing, and actually improving, its time performance.
Our experimental results show that the wavelet matrix outperforms all the
wavelet tree variants along the space/time tradeoff map.
\end{abstract}

\section{Introduction}

In many applications related to text indexing and succinct data structures,
it is necessary to represent a sequence $S[1,n]$ over an integer
alphabet $[0,\sigma)$ so as to support the following functionality:

\begin{itemize}
\item $\access(S,i)$ returns $S[i]$.
\item $\rank_a(S,i)$ returns the number of occurrences of symbol $a$ in
$S[1,i]$.
\item $\select_a(S,j)$ returns the position in $S$ of the $j$-th occurrence of
symbol $a$.
\end{itemize}

Some examples where this problem arises are 
indexes for supporting indexed pattern matching on strings 
	\cite{GGV03,GV05,FM05,FMMN07,NM07}, 
indexes for solving computational biology problems on sequences 
	\cite{SOG10,BGOS11},
simulation of inverted indexes over natural language text collections 
	\cite{BFLN12,AGO10}, 
representation of labeled trees and XML structures 
	\cite{BDMRRR05,BCAHM07,FLMM09,ACMNNSV10,BHMR11},
representation of binary relations and graphs
	\cite{BGMR07,CN10,BCN10,BHMR11},
solving document retrieval problems
	\cite{VM07,GNP11},
and many more.

An elegant data structure to solve this problem is the {\em wavelet tree}
\cite{GGV03}. In its most basic form, this is a 
balanced tree of $O(\sigma)$ nodes storing bitmaps. It requires
$n\lg\sigma + o(n\lg\sigma) + O(\sigma\lg n)$ bits to represent $S$ and 
solves the three queries in time $O(\lg\sigma)$. The wavelet tree supports
not only the three queries we have mentioned, but more general
range search operations that find applications in representing geometric
grids \cite{Cha88,BHMM09,BLNS10,BCN10,NNR13} and text indexes based on them
\cite{Nav04,FM05,MN07,CHSV08,CNfi10,KN12,NN12}, complex queries on numeric 
sequences \cite{GPT09,KN12,GKNP12}, and many others. Various recent surveys
\cite{NM07,FGM09,GVX11,Mak12,Nav12} are dedicated, partially or totally, to 
the number of applications of this versatile data structure.

In various applications, the alphabet size $\sigma$ is significant compared
to the length $n$ of the sequence. Some examples are sequences of words 
(seen as integer tokens) when simulating inverted indexes, sequences of XML
tags, and sequences of document numbers in document retrieval. When using
wavelet trees to represent grids, the sequence length $n$ becomes the width
of the grid and the alphabet size becomes the height of the grid, and both
are equal in most cases.

A large value of $\sigma$ affects the space usage of wavelet trees.
A pointerless wavelet tree 
\cite{MN07} concatenates all the bitmaps levelwise and removes the 
$O(\sigma \lg n)$ bits from the space. It retains the time complexity of
pointer-based wavelet trees, albeit it is slower in practice. This 
representation can be made to use $nH_0(S) + o(n\lg\sigma)$ bits, where
$H_0(S) \le \lg\sigma$ is the per-symbol zero-order entropy of $S$,
by using compressed bitmaps \cite{RRR07,GGV03}.
This makes the wavelet tree traversal even slower in practice, however.
 
A pointer-based wavelet tree, instead, can achieve zero-order compression by
replacing the balanced tree by the Huffman tree \cite{Huf52}. Then, even
without compressing the bitmaps, the storage space becomes 
$n(H_0(S)+1) + o(n(H_0(S)+1)) + O(\sigma\lg n)$ bits. Adding bitmap compression
removes the $n$ bits of the Huffman redundancy.
In addition, this technique is faster than the basic one, as the average access
time is $O(H_0(S))$. However, it still requires the $O(\sigma\lg n)$ extra
bits.

Other than wavelet trees, Golynski et al.~\cite{GMR06} proposed a sequence
representation for large alphabets, which uses $n\lg\sigma + o(n\lg\sigma)$ 
bits (no compression) and offers much faster time complexities to support the
three operations, $O(\lg\lg\sigma)$. Later, Barbay et al.~\cite{BGNN10} built on
this idea to obtain zero-order compression, $nH_0(S)+o(n(H_0(S)+1))$ bits,
while retaining the times. This so-called ``alphabet-partitioned'' 
representation does not, however, offer the 
richer functionality of wavelet trees. Moreover, as shown in their experiments
\cite{BCGNN13}, its sublinear space terms are higher in practice than those of
a zero-order compressed wavelet tree (yet their better complexity does show up
in practice). There are recent theoretical developments slightly improving those
complexities \cite{BN12}, but their sublinear space terms would be even higher 
in practice.

\paragraph{Our contribution.}

In this article we introduce the {\em
wavelet matrix}. This is an alternative representation of the balanced
pointerless wavelet tree that reorders the nodes in each level, in a way that 
retains all the wavelet tree functionality while the traversals needed to
carry out the operations are simplified and sped up. The wavelet matrix then
retains all the capabilities of wavelet trees, is resistant to large alphabets,
and its speed gets close to that of pointer-based wavelet trees. It can also
obtain zero-order compression by compressing the bitmaps (which slows it down).

We then consider how to give Huffman shape to the wavelet trees without the
burden of storing the tree pointers. This is achieved by combining canonical
Huffman codes \cite{SK64} with pointerless wavelet trees (now unbalanced).
Finally, we aim at combining both improvements, that is, obtaining 
Huffman shaped wavelet matrices. These yield simultaneously zero-order
compression and fast operations. It turns out, however, that the canonical 
Huffman
codes cannot be directly combined with the node numbering induced by the 
wavelet matrix, so we derive an alternative code assignment scheme that is
also optimal and compatible with the wavelet matrix.

We implement all the variants and test them over various real-life sequences, 
showing that a few versions of the wavelet matrix dominate all the wavelet tree
variants across the space/time tradeoff map, on diverse sequences over large 
alphabets and point grids.

\section{Basic Concepts}

\subsection{Wavelet Trees}
\label{sec:struc}

A wavelet tree \cite{GGV03} for sequence $S[1,n]$ over alphabet $[0..\sigma)$
is a complete balanced binary tree, where each node handles a range of
symbols. The root handles $[0..\sigma)$ and each leaf handles one symbol. Each
node $v$ handling the range $[\alpha_v,\omega_v)$ represents the subsequence
$S_v[1,n_v]$ of $S$ formed by the symbols in $[\alpha_v,\omega_v)$, but it
does not explicitly store $S_v$. Rather, internal nodes $v$ store a bitmap 
$B_v[1,n_v]$, so that $B_v[i]=0$ if $S_v[i] < \alpha_v+2^{\lceil
\lg(\omega_v-\alpha_v)\rceil-1}$ and $B_v[i]=1$ otherwise. That is, we
partition the alphabet interval $[\alpha_v,\omega_v)$ into two roughly equal
parts: a ``left'' one, $[\alpha_v, \alpha_v+2^{\lceil
\lg(\omega_v-\alpha_v)\rceil-1})$ and a ``right'' one, $[\alpha_v+2^{\lceil
\lg(\omega_v-\alpha_v)\rceil-1}, \omega_v)$. These are handled by the left and
right children of $v$. No bitmaps are stored for the leaves.
Figure~\ref{fig:wtree} (left) gives an example.

\begin{figure}[t]
\hspace*{-2mm}\includegraphics[height=4.4cm]{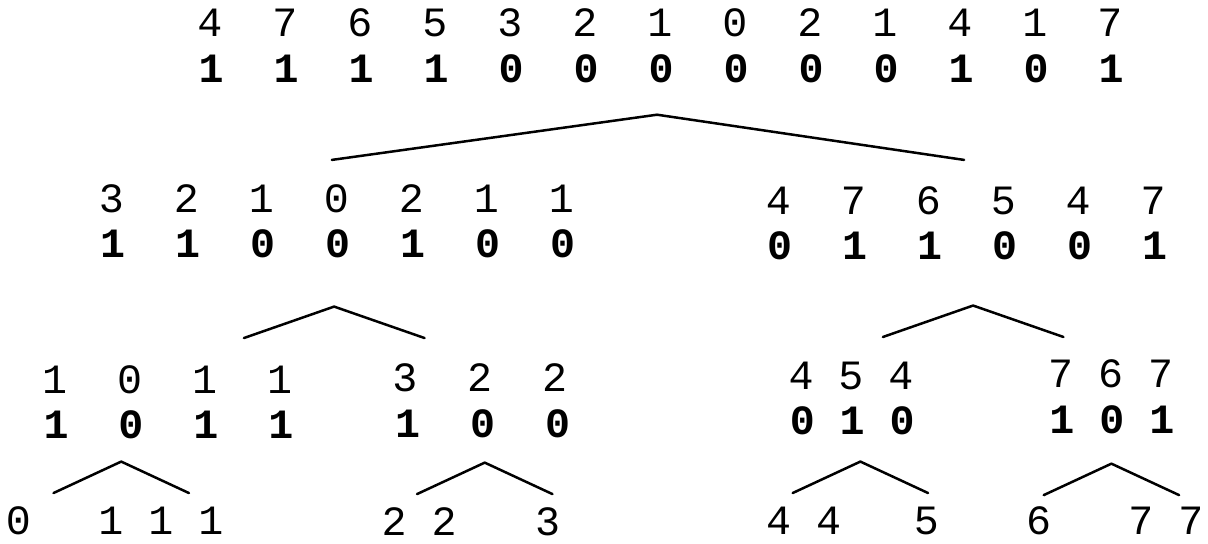}   \hfill

\vspace*{-4cm}
\hfill \includegraphics[height=3.6cm]{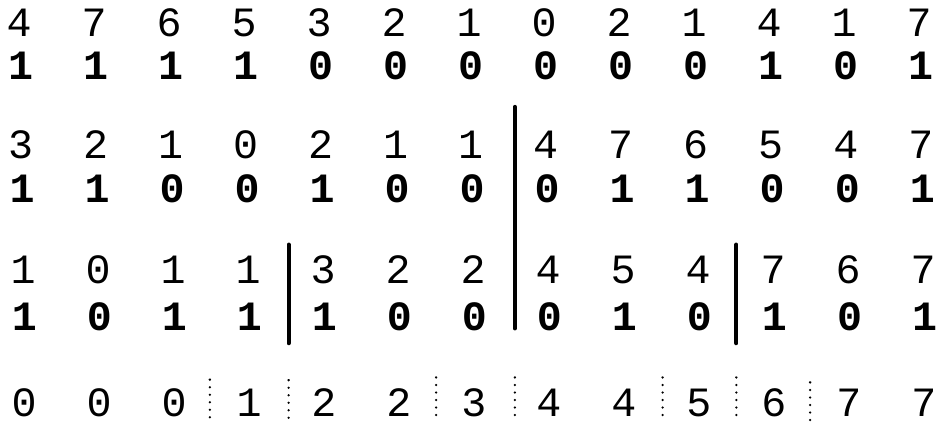} \hspace*{-4mm}

\caption{On the left, the standard wavelet tree over a sequence. The
subsequences $S_v$ are not stored. The bitmaps $B_v$,
in bold, are stored, as well as the tree topology. On the right, its pointerless
version. The divisions into nodes are not stored but computed on the fly.}
\label{fig:wtree}
\end{figure}

The tree has height $\lceil \lg \sigma \rceil$, and it has exactly $\sigma$
leaves and $\sigma-1$ internal nodes. If we regard it level by level, we can
see that it holds, in the $B_v$ bitmaps, exactly $n$ bits per level (the
lowest one may hold fewer bits). Thus it stores at most
$n\lceil\lg\sigma\rceil$ bits. Storing the tree pointers, and pointers to the
bitmaps, requires $O(\sigma \lg n)$ further bits, if we use the minimum of
$\lg n$ bits for the pointers.

To access $S[i]$, we start from the root node $\nu$, setting $i_\nu=i$. 
If $B_{\nu}[i_\nu]=0$, this means that $S[i] = S_\nu[i_\nu] < 
2^{\lceil \lg \sigma\rceil-1}$ and that the symbol is represented in the 
subsequence $S_{\nu_l}$ of the left child $\nu_l$ of the root. Otherwise, 
$S_\nu[i_\nu] \ge 2^{\lceil \lg \sigma\rceil-1}$ and it is represented in the 
subsequence $S_{\nu_r}$ of the right child $\nu_r$ of the root. In the first 
case, the position of $S_\nu[i_\nu]$ in $S_{\nu_l}$ is $i_{\nu_l} = 
\rank_0(B_\nu,i_\nu)$, whereas in the second, the position in $S_{\nu_r}$ is 
$i_{\nu_r} = \rank_1(B_\nu,i_\nu)$. We continue recursively, extracting 
$S_v[i_v]$ from node $v=\nu_l$ or $v=\nu_r$, until we arrive at a leaf 
representing the alphabet interval $[a,a]$, where we can finally report 
$S[i] = a$.

Therefore, the cost of operation $\access$ is that of $\lceil\lg\sigma\rceil$
binary $\rank$ operations on bitmaps $B_v$. Binary $\rank$ and $\select$
operations can be carried out in constant time using only $o(n_v)$ bits on top
of $B_v$ \cite{Jac89,Mun96,Cla96}.

The process to compute $\rank_a(S,i)$ is similar. The difference is that we do
not descend according to whether $B_v[i]$ equals $0$ or $1$, but rather
according to the bits of $a \in [0,\sigma)$: the highest bit of $a$ tells us 
whether to go left or right, and the lower bits are used in the next levels. 
When moving to a child $u$ of $v$, we compute $i_u=\rank_{0/1}(B_v,i_v)$ to be 
the number of times the current bit of $a$ appears in $B_v[1,i_v]$. When we 
arrive at the leaf $u$ handling the range $[a,a]$, the answer to $\rank_a(S,i)$
is $i_u$.

Finally, to compute $\select_a(S,j)$ we must proceed upwards. We start at the
leaf $u$ that handles the alphabet range $[a,a]$. So we want to track the
position of $S_u[j_u]$, $j_u=j$, towards the root. If $u$ is the left child of
its parent $v$, then the corresponding position at the parent is $S_v[j_v]$,
where $j_v = \select_0(B_v,j_u)$. Else, the corresponding position is
$j_v = \select_1(B_v,j_u)$. When we finally arrive at the root $\nu$,
the answer to the query is $j_\nu$.

Thus the cost of query $\rank_a(S,i)$ is $\lceil \lg\sigma \rceil$ binary
$\rank$ operations (just like $\access(S,i)$), and the cost of query
$\select_a(S,i)$ is $\lceil \lg\sigma \rceil$ binary $\select$ operations.
Algorithm~\ref{alg:standard} gives the pseudocode (the recursive form is
cleaner, but recursion can be easily removed).

\begin{algorithm}[t]
\caption{Standard wavelet tree algorithms: On the wavelet tree of sequence $S$
rooted at $\nu$,
$\mathbf{acc}(\nu,i)$ returns $S[i]$;
$\mathbf{rnk}(\nu,a,i)$ returns $\rank_a(S,i)$; and
$\mathbf{sel}(\nu,a,j)$ returns $\select_a(S,j)$.
The left/right children of $v$ are called $v_l/v_r$.}
\label{alg:standard}
\begin{tabular}{ccc}
\begin{minipage}{0.30\textwidth}

$\mathbf{acc}(v,i)$

\vspace{-0.1cm}
\begin{algorithmic}
\IF{$\omega_v-\alpha_v=1$}
   \RET $\alpha_v$
\ENDIF
\IF{$B_v[i]=0$}
   \STATE $i \leftarrow \rank_0(B_v,i)$
   \RET $\mathbf{acc}(v_l,i)$
\ELSE
   \STATE $i \leftarrow \rank_1(B_v,i)$
   \RET $\mathbf{acc}(v_r,i)$
\ENDIF
\end{algorithmic}
\end{minipage}
&
\begin{minipage}{0.33\textwidth}
$\mathbf{rnk}(v,a,i)$

\vspace{-0.1cm}
\begin{algorithmic}
\IF{$\omega_v-\alpha_v=1$}
   \RET $i$
\ENDIF
\IF{$a < 2^{\lceil \lg(\omega_v-\alpha_v)\rceil-1}$}
   \STATE $i \leftarrow \rank_0(B_v,i)$
   \RET $\mathbf{rnk}(v_l,a,i)$
\ELSE
   \STATE $i \leftarrow \rank_1(B_v,i)$
   \RET $\mathbf{rnk}(v_r,a,i)$
\ENDIF
\end{algorithmic}
\end{minipage}
&
\begin{minipage}{0.33\textwidth}
$\mathbf{sel}(v,a,j)$

\vspace{-0.1cm}
\begin{algorithmic}
\IF{$\omega_v-\alpha_v=1$}
   \RET $j$
\ENDIF
\IF{$a < 2^{\lceil \lg(\omega_v-\alpha_v)\rceil-1}$}
   \STATE $j \leftarrow \mathbf{sel}(v_l,a,j)$
   \RET $\select_0(B_v,j)$
\ELSE
   \STATE $j \leftarrow \mathbf{sel}(v_r,a,j)$
   \RET $\select_1(B_v,j)$
\ENDIF
\end{algorithmic}
\end{minipage}
\end{tabular}
\end{algorithm}

\subsection{Pointerless Wavelet Trees}

Since the wavelet tree is a complete balanced binary tree, it is possible to
concatenate all the bitmaps at each level and still retain the same
functionality \cite{MN07}. Instead of a bitmap per node $v$, there will
be a single bitmap per level $\ell$, $\B_\ell[1,n]$. Figure~\ref{fig:wtree}
(right) illustrates this arrangement. The main problem is how to keep
track of the range $\B_\ell[s_v,e_v]$ corresponding to a node $v$ of depth
$\ell$.

\paragraph{The strict variant.}

The strict variant \cite{MN07} stores no data apart from the
$\lceil\lg\sigma\rceil$
pointers to the level bitmaps. Keeping track of the node ranges
is not hard if we start at the root (as in $\access$ and $\rank$).
Initially, we know that $[s_\nu,e_\nu] = [1,n]$, that is, the whole bitmap
$\B_0$ is equal to the bitmap of the root, $B_\nu$. Now, imagine that we
have navigated towards a node $v$ at depth $\ell$, and know $[s_v,e_v]$.
The two children of $v$ share the same interval $[s_v,e_v]$ at $\B_{\ell+1}$.
The split point is $m=\rank_0(\B_\ell,e_v) - \rank_0(\B_\ell,s_v-1)$, the number
of 0s in $\B_\ell[s_v,e_v]$. Then, if we descend to the left child $v_l$, we
will have $[s_{v_l},e_{v_l}] = [s_v,s_v+m-1]$. If we descend to the right child
$v_r$, we will have $[s_{v_r},e_{v_r}] = [s_v+m,e_v]$.

Things are a little bit harder for $\select$, because we must proceed upwards.
In the strict variant, the way to carry out $\select_a(S,j)$ is to first
descend to the leaf corresponding to symbol $a$, and then track the leaf
position $j$ up to the root as we return from the recursion.

Algorithm~\ref{alg:levelwise} gives the pseudocode (we use $p=s-1$ instead of
$s=s_v$).
Note that, compared to the standard version, the strict variant requires two
extra binary $\rank$ operations per original binary $\rank$, on the top-down
traversals (i.e., for queries $\access$ and $\rank$). For query $\select$, the 
strict variant requires two extra binary $\rank$ operations per original binary
$\select$. Thus the times may up to triple for these traversals.\footnote{In 
practice the effect is not so large because of cache effects when $s_v$ is 
close to $e_v$. In addition, binary $\select$ is more expensive than $\rank$
in practice, thus the impact on query $\select$ is lower.}

\begin{algorithm}[t]
\caption{Pointerless wavelet tree algorithms (strict variant): On the wavelet
tree of sequence $S$,
$\mathbf{acc}(0,i,0,n)$ returns $S[i]$;
$\mathbf{rnk}(0,a,i,0,n)$ returns $\rank_a(S,i)$; and
$\mathbf{sel}(0,a,j,0,n)$ returns $\select_a(S,j)$.
For simplicity we have omitted the computation of $[\alpha_v,\omega_v)$.}
\label{alg:levelwise}
\begin{tabular}{ccc}
\begin{minipage}{0.27\textwidth}
$\mathbf{acc}(\ell,i,p,e)$
\begin{algorithmic}
\vspace*{-4mm}
\IF{$\omega_v-\alpha_v=1$}
   \RET $\alpha_v$
\ENDIF
\STATE $l \leftarrow \rank_0(\B_\ell,p)$
\STATE $r \leftarrow \rank_0(\B_\ell,e)$
\IF{$\B_\ell[i]=0$}
   \STATE $z \leftarrow \rank_0(\B_\ell,p+i)$
   \RET $\mathbf{acc}(\ell{+}1,~~~$ $~~~z{-}l,p,p{+}r{-}l)$
\ELSE
   \STATE $z \leftarrow \rank_1(\B_\ell,p+i)$
   \RET $\mathbf{acc}(\ell{+}1,~~~$ $~~~z{-}(p{-}l),p{+}r{-}l,e)$
\ENDIF
\end{algorithmic}
\end{minipage}
&
\begin{minipage}{0.29\textwidth}
$\mathbf{rnk}(\ell,a,i,p,e)$
\begin{algorithmic}
\vspace*{-4mm}
\IF{$\omega_v-\alpha_v=1$}
   \RET $i$
\ENDIF
\STATE $l \leftarrow \rank_0(\B_\ell,p)$
\STATE $r \leftarrow \rank_0(\B_\ell,e)$
\IF{$a < 2^{\lceil \lg(\omega_v-\alpha_v)\rceil-1}$}
   \STATE $z \leftarrow \rank_0(\B_\ell,p+i)$
   \RET $\mathbf{rnk}(\ell{+}1,a,~~~~$ $~~~z{-}l,p,p{+}r{-}l)$
\ELSE
   \STATE $z \leftarrow \rank_1(\B_\ell,p+i)$
   \RET $\mathbf{rnk}(\ell{+}1,a,~~~~$ $~~~z{-}(p{-}l),p{+}r{-}l,e)$
\ENDIF
\end{algorithmic}
\end{minipage}
&
\begin{minipage}{0.36\textwidth}
$\mathbf{sel}(\ell,a,j,p,e)$
\begin{algorithmic}
\IF{$\omega_v-\alpha_v=1$}
   \RET $j$
\ENDIF
\STATE $l \leftarrow \rank_0(\B_\ell,p)$
\STATE $r \leftarrow \rank_0(\B_\ell,e)$
\IF{$a < 2^{\lceil \lg(\omega_v-\alpha_v)\rceil-1}$}
   \STATE $j {\leftarrow} \mathbf{sel}(\ell{+}1,a,j,p,p{+}r{-}l)$
   \RET $\select_0(\B_\ell,l+j){-}p$
\ELSE
   \STATE $j {\leftarrow} \mathbf{sel}(\ell{+}1,a,j,p{+}r{-}l,e)$
   \RET $\select_1(\B_\ell,(p{-}l){+}j){-}p$
\ENDIF
\end{algorithmic}
\ \\
\end{minipage}
\end{tabular}
\end{algorithm}

\paragraph{The extended variant.}

The {\em extended} variant \cite{CN08}, instead, stores an array
$C[0,\sigma-1]$ of pointers to the $\sigma$ starting positions of the symbols
in the (virtual) array of the leaves, or said another way, $C[a]$ is the
number of occurrences of symbols smaller than $a$ in $S$. Note this array
requires $O(\sigma\lg n)$ bits (or at best $O(\sigma \lg(n/\sigma))+o(n)$ if
represented as a compressed bitmap \cite{RRR07}), but the constant is much
lower than on a pointer-based tree (which stores the left child, the right 
child, the parent, the value $n_v$, the pointer to bitmap $B_v$, pointers to
the leaves, etc.).

With the help of array $C$, the number of operations becomes closer to the
standard  version, since $C$ lets us compute the ranges: The range of
any node $v$ is simply $[s_v,e_v]=[C[\alpha_v]+1,C[\omega_v]]$. In the 
algorithms for queries $\access$ and $\rank$, where we descend from the root, 
the values $\alpha_v$ and $\omega_v$ are easily maintained. Thus we do not
need to compute $r$ in Algorithm~\ref{alg:levelwise}, as it is used only to
compute $e=e_v= C[\omega_v]$. Thus we require only one extra binary $\rank$
operation per level.

This is slightly more complicated when solving query $\select_a(S,j)$.
We start at offset $j$ in the interval $[C[\alpha_u]+1,C[\omega_u]]$ for
$(\alpha_u,\omega_u) = (a,a+1)$ and track this
position upwards: If the leaf $u$ is a left child of its parent $v$ (i.e.,
if $\alpha_u$ is even), then the parent's range (in the deepest bitmap 
$\B_\ell$) is $(\alpha_v,\omega_v) = (\alpha_u,\omega_u+1)$. Instead, if the 
leaf is a right child of its parent, then the parent's range is
$(\alpha_v,\omega_v) = (\alpha_u-1,\omega_u)$. We use binary $\select$ on the
range $[C[\alpha_v]+1,C[\omega_v]]$ to map the position $j$ to the parent's 
range. Now we proceed similarly at the parent $w$ of $v$. If
$\alpha_v = 0~\mathrm{mod}~4$, then $v$ is the left child of $w$, otherwise it 
is the right child. In the first case, the range of $w$ in bitmap
$\B_{\ell-1}$ is $(\alpha_w,\omega_w) = (\alpha_v,\omega_v+2)$, otherwise it is
$(\alpha_w,\omega_w) = (\alpha_v-2,\omega_v)$.
We continue until the root, where $j$ is the answer. In this case we need only
one extra binary $\rank$ operation per level.
Algorithm~\ref{alg:levelwiseC} details the algorithms.

\begin{algorithm}[t]
\caption{Pointerless wavelet tree algorithms (extended variant): On the wavelet
tree of sequence $S$,
$\mathbf{acc}(0,i)$ returns $S[i]$;
$\mathbf{rnk}(0,a,i)$ returns $\rank_a(S,i)$; and
$\mathbf{sel}(a,j)$ returns $\select_a(S,j)$.
For simplicity we have omitted the computation of $[\alpha_v,\omega_v)$,
except on $\mathbf{sel}(a,j)$, where for simplicity we assume $C[a]$ refers
to level $\ell = \lceil\lg\sigma\rceil$, where in fact it could refer to
level $\ell = \lceil\lg\sigma\rceil-1$.}
\label{alg:levelwiseC}
\begin{tabular}{ccc}
\begin{minipage}{0.32\textwidth}
$\mathbf{acc}(\ell,i)$
\begin{algorithmic}
\IF{$\omega_v-\alpha_v=1$}
   \RET $\alpha_v$
\ENDIF
\STATE $l \leftarrow \rank_0(\B_\ell,C[\alpha_v])$
\STATE $z \leftarrow \rank_0(\B_\ell,C[\alpha_v]{+}i)$
\IF{$\B_\ell[i]=0$}
   \RET $\mathbf{acc}(\ell{+}1,z{-}l)$
\ELSE
   \RET $\mathbf{acc}(\ell{+}1,i{-}(z{-}l))$
\ENDIF
\ \\
\ \\
\ \\
\ \\
\ \\
\end{algorithmic}
\end{minipage}
&
\begin{minipage}{0.34\textwidth}
$\mathbf{rnk}(\ell,a,i)$
\begin{algorithmic}
\IF{$\omega_v-\alpha_v=1$}
   \RET $i$
\ENDIF
\STATE $l \leftarrow \rank_0(\B_\ell,C[\alpha_v])$
\STATE $z \leftarrow \rank_0(\B_\ell,C[\alpha_v]+i)$
\IF{$a < 2^{\lceil \lg(\omega_v-\alpha_v)\rceil-1}$}
   \RET $\mathbf{rnk}(\ell{+}1,a,z{-}l)$
\ELSE
   \RET $\mathbf{rnk}(\ell{+}1,a,i{-}(z{-}l))$
\ENDIF
\ \\
\ \\
\ \\
\ \\
\ \\
\end{algorithmic}
\end{minipage}
&
\begin{minipage}{0.36\textwidth}
$\mathbf{sel}(a,j)$
\begin{algorithmic}
\STATE $\ell \leftarrow \lceil \lg \sigma \rceil$, 
       $d \leftarrow 1$
\WHILE{$\ell \ge 0$}
   \IF{$a ~\mathrm{mod}~ 2^d = 0$}
      \STATE $l \leftarrow \rank_0(\B_\ell,C[\alpha_v])$
      \STATE $j \leftarrow \select_0(\B_\ell,l{+}j)$
   \ELSE
      \STATE $\alpha_v \leftarrow \alpha_v-2^{d-1}$
      \STATE $l \leftarrow \rank_1(\B_\ell,C[\alpha_v])$
      \STATE $j \leftarrow \select_1(\B_\ell,l{+}j)$
   \ENDIF
   \STATE $j \leftarrow j - C[\alpha_v]$
   \STATE $\ell \leftarrow \ell-1$,
          $d \leftarrow d+1$
\ENDWHILE
\RET $j$
\end{algorithmic}
\end{minipage}
\end{tabular}
\end{algorithm}

\subsection{Huffman Shaped Wavelet Trees}

Given the frequencies of the $\sigma$ symbols in $S[1,n]$, the Huffman 
algorithm \cite{Huf52} produces an optimal variable-length encoding so that 
(1) it is prefix-free, that is, no code is a prefix of another; 
(2) the size of the compressed sequence is minimized. If symbol
$a \in [0,\sigma)$ appears $n_a$ times in $S$, then the Huffman algorithm
will assign it a codeword of length $\ell_a$ so that the sum
$L=\sum_a n_a \ell_a$ is minimized. Then the file is compressed to $L$ bits by
replacing each symbol $S[i]=a$ by its code of length $\ell_a$. 
The {\em empirical zero-order entropy} \cite{CT91}
of $S$ is $H_0(S) = \sum_a \frac{n_a}{n}\lg\frac{n}{n_a}\le \lg\sigma$, and 
no statistical compressor based on individual symbol probabilities can output 
less than $nH_0(S)$ bits. The output size of Huffman compression can be
bounded by $\sum_a n_a \ell_a < n(H_0(S)+1)$ bits, which is off the optimum
by less than 1 bit per symbol.

Huffman \cite{Huf52} showed how to build a so-called {\em Huffman tree} to obtain these codes.
The tree leaves will contain the symbols, whose codes are obtained by following
the tree path from the root to their leaves. Each branch of the tree is labeled
with 0 (say, the left child) or 1 (say, the right child), and the code 
associated with a symbol $a$ is obtained by concatenating the labels found in 
the path from the tree root to the leaf that contains symbol $a$.  

Building a balanced wavelet tree is equivalent to using a fixed-length encoding
of $\lceil \lg \sigma \rceil$ or $\lfloor \lg\sigma \rfloor$
bits per symbol. Instead, by giving the wavelet
tree the shape of the Huffman tree, the total number of bits stored is exactly 
the output size of the Huffman compressor \cite{GGV03,Nav12}: The leaf of $a$
is at depth $\ell_a$, and each of the $n_a$ occurrences induces one bit in the
bitmap of each of the $\ell_a$ ancestors of the leaf. The size 
of this tree, plus $\rank$/$\select$ overheads, is thus upper bounded by 
$n(H_0(S)+1)+o(n(H_0(S)+1))+O(\sigma \lg n)$ bits.
Figure~\ref{fig:huffman} depicts a Huffman shaped wavelet tree.

\begin{figure}[t]
  \hspace*{-2mm}\includegraphics[height=4.3cm]{wt.pdf}   \hfill

 \vspace*{-4.2cm}

 \hfill \includegraphics[height=4.3cm]{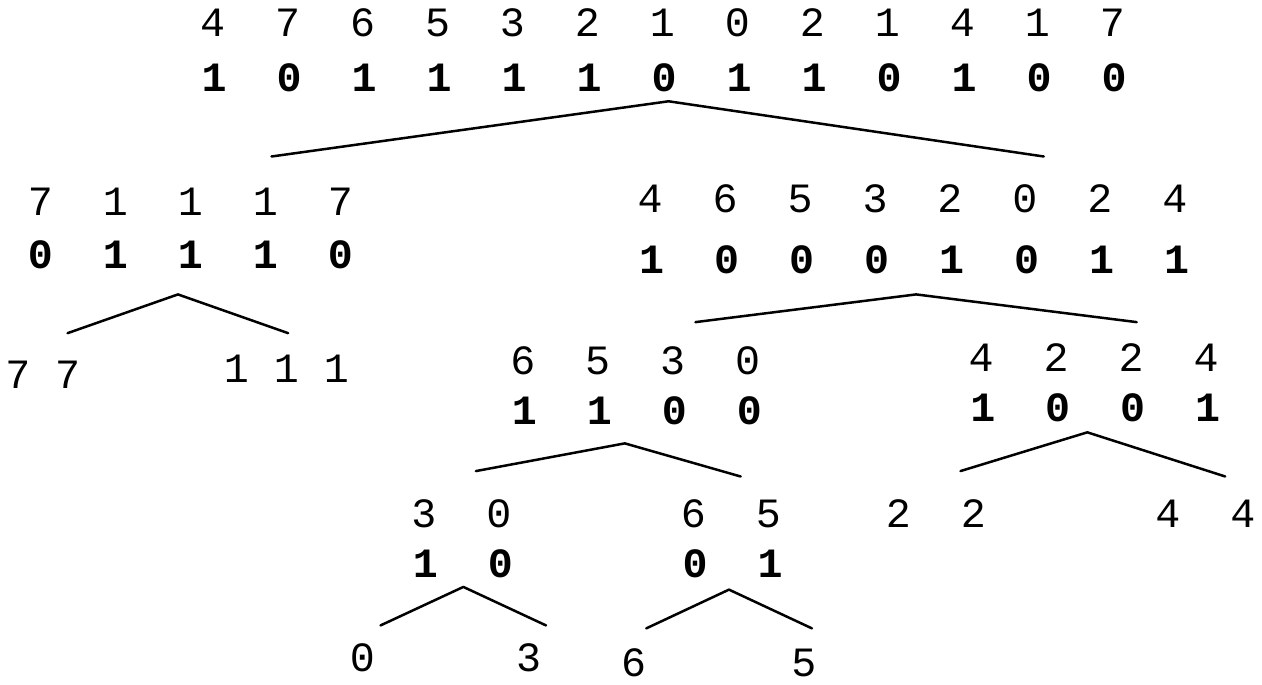} \hspace*{-2mm}

\caption{On the left, the same wavelet tree of Figure~\ref{fig:wtree}. On the
right, its Huffman shaped version.}
\label{fig:huffman}
\end{figure}

The wavelet tree operations are performed verbatim on Huffman shaped wavelet
trees. Moreover, they become faster on average: If $i \in [1,n]$ is chosen
at random for $\access(S,i)$, or $a$ is chosen with probability $n_a/n$
in operations $\rank_a(S,i)$ and $\select_a(S,j)$ (which is the typical case
in most applications), then the average time is $O(H_0(S)+1)$. By rebalancing 
deep leaves, the space and average time are maintained and the worst-case time 
of the operations is limited to $O(\lg\sigma)$ \cite{BN11}.

Zero-order compression can also be achieved on the balanced wavelet tree,
by using a compressed representation of the bitmaps \cite{RRR07}. The time
remains the same and the space decreases to $nH_0(S)+o(n\lg\sigma)$ bits
\cite{GGV03}. Combining the compressed bitmap representation with Huffman shape,
we obtain $nH_0(S) + o(n(H_0(S)+1)) + O(\sigma\lg n)$ bits. This combination
works well in practice \cite{CN08}, although the compressed bitmap 
representation is in practice slower than the plain one.

\subsection{Wavelet Trees on Point Grids}

As mentioned in the Introduction, wavelet trees are not only useful to 
support $\access$, $\rank$ and $\select$ operations on sequences. They are
also frequently used to represent point grids \cite{Cha88,Nav12}, where they
can for example count or list the points that lie in a rectangular area.
Typically the grid is square, of $n \times n$ cells, and contains $n$ points,
exactly one per point and per column (other arrangements are routinely mapped
to this simplified case). Then it can be regarded as a sequence $S[1,n]$ over
a large alphabet of size $\sigma=n$. In this case, pointer-based wavelet trees
perform poorly, as the space for the pointers is dominant. Similarly, 
zero-order compression is ineffective. The balanced wavelet 
trees without pointers \cite{MN07} are the most successful representation.

The pseudocode for range searches using standard wavelet trees is easily
available, see for example Gagie et al.~\cite{GNP11}. Algorithm~\ref{alg:range} 
shows the algorithms adapted to pointerless wavelet trees. We consider the two
basic operations $\rc(P,x_1,x_2,y_1,y_2)$ and $\rr(P,x_1,x_2,y_1,y_2)$, which
count and list, respectively, the points within the rectangle $[x_1,x_2]
\times [y_1,y_2]$ from the point set represented in sequence $P[1,n]$. The time
complexities can be shown to be $O(\lg n)$ for $\rc$ and $O(k\lg(n/k))$ for a
$\rr$ operation that lists $k$ points. In practical terms, compared to the
standard versions, the pointerless algorithms requires twice the number of
$\rank$ operations.

\begin{algorithm}[t]
\caption{Range search algorithms on pointerless wavelet trees: 
$\mathbf{count}(0,x_1,x_2,y_1,y_2,0,n)$
returns $\rc(P,x_1,x_2,y_1,y_2)$ on the wavelet tree of sequence $P$; and
$\mathbf{report}(0,x_1,x_2,y_1,y_2,0,n)$ outputs all those $y$, where a point
with coordinate $y_1 \le y \le y_2$ appears in $P[x_1,x_2]$. For
simplicity we have omitted the computation of $[\alpha_v,\omega_v)$.}
\label{alg:range}
\begin{tabular}{cc}
\begin{minipage}{0.5\textwidth}
$\mathbf{count}(\ell,x_1,x_2,y_1,y_2,p,e)$
\begin{algorithmic}
\IF{$x_1 > x_2 ~\lor~ [\alpha_v,\omega_v] \cap [y_1,y_2] = \emptyset$}
        \RET 0
\ELSIF{$[\alpha_v,\omega_v] \subseteq [y_1,y_2]$}
        \RET $x_2-x_1+1$
\ELSE
        \STATE $l \leftarrow \rank_0(\B_\ell,p)$
        \STATE $r \leftarrow \rank_0(\B_\ell,e)$
        \STATE $x_1^l \leftarrow \rank_0(\B_\ell,x_1-1)-l+1$
        \STATE $x_2^l \leftarrow \rank_0(\B_\ell,x_2)-l$
        \STATE $x_1^r \leftarrow x_1-x_1^l+1$, $x_2^r \leftarrow x_2-x_2^l$
        \RET $\mathbf{count}(\ell{+}1,x_1^l,x_2^l,y_1,y_2,p,p{+}r{-}l)$
        \STATE \hspace{0.97cm} $+ \mathbf{count}(\ell{+}1,x_1^r,x_2^r,y_1,y_2,
						p{+}r{-}l,e)$
\ENDIF
\end{algorithmic}
\end{minipage}
&
\begin{minipage}{0.5\textwidth}
$\mathbf{report}(v,x_1,x_2,y_1,y_2,p,e)$
\begin{algorithmic}
\IF{$x_1 > x_2 ~\lor~ [\alpha_v,\omega_v] \cap [y_1,y_2] = \emptyset$}
        \RET
\ELSIF{$\omega_v-\alpha_v=1$}
        \OUTPUT $\alpha_v$
\ELSE
        \STATE $l \leftarrow \rank_0(\B_\ell,p)$
        \STATE $r \leftarrow \rank_0(\B_\ell,e)$
        \STATE $x_1^l \leftarrow \rank_0(\B_\ell,x_1-1)-l+1$
        \STATE $x_2^l \leftarrow \rank_0(\B_\ell,x_2)-l$
        \STATE $x_1^r \leftarrow x_1-x_1^l+1$, $x_2^r \leftarrow x_2-x_2^l$
        \STATE $\mathbf{report}(\ell{+}1,x_1^l,x_2^l,y_1,y_2,p,p{+}r{-}l)$
        \STATE $\mathbf{report}(\ell{+}1,x_1^r,x_2^r,y_1,y_2,p{+}r{-}l,e)$
\ENDIF
\end{algorithmic}
\end{minipage}
\end{tabular}
\end{algorithm}

\section{Pointerless Huffman Shaped Wavelet Trees}
\label{sec:levelhuff}

In this section we show how to use canonical Huffman codes \cite{SK64} to
represent Huffman shaped wavelet trees without pointers, this way removing the
main component of the $O(\sigma\lg n)$ extra bits and retaining the
advantages of reduced space and $O(H_0(S)+1)$ average traversal time.

The problem that arises when storing a standard Huffman shaped wavelet tree
in levelwise form is that a leaf that appears in the middle of a level leaves
a ``hole'' that ruins the calculations done at the nodes to the right of it
to find their position in the next level. Canonical Huffman codes choose one
of the many optimal Huffman trees that, among other interesting benefits
\cite{SK64,Sal07}, yields a set of codes in which longer codes appear to the
left of shorter codes.\footnote{It is usually to the right, but this way is more
convenient for us.} As a consequence, all the leaves of a level appear grouped
to the right, and therefore do not alter the navigation calculations for the
other nodes. The levelwise deployment of the tree can be seen as a sequence of
``contiguous'' bitmaps of varying length. 

The navigation procedures of Algorithm~\ref{alg:levelwise} can then be used 
verbatim, except for a few alphabet mappings that must be carried out: For
$\access(S,i)$, we need to maintain the Huffman tree so that, given the 0/1 
labels of the traversed path, we determine the alphabet symbol corresponding 
to that leaf of the Huffman tree. For $\rank_a(S,i)$, we need to convert the
symbol $a$ to its variable-length code, in order to follow the corresponding
path in the wavelet tree. Finally, for $\select_a(S,i)$, we need the same as
for $\rank$ for the strict variant, or a pointer to the corresponding leaf area
in some bitmap $\B_\ell$, for the extended variant. The mappings are also used
to determine when to stop a top-down traversal. The mapping information amounts
to $O(\sigma\lg n)$ bits as well, but it is much less in practice than what is
stored for pointer-based wavelet trees, as explained. Moreover, in the case of
canonical codes, $\sigma\lg\sigma+O(\sigma)$ bits are sufficient to represent
the mappings. It has also been shown that they can be represented 
within $O(\sigma\lg\lg n)$ bits as well \cite{NO13}.

The maximum number of levels in a Huffman tree is $O(\lg n)$, and as explained
it can be made $O(\lg\sigma)$ without affecting the asymptotic performance.
Thus the pointers to the levels add up to a negligible $O(\lg^2 n)$ bits.
The rest of the space is as for standard Huffman shaped wavelet trees:
$n(H_0(S)+1)+o(n(H_0(S)+1))$ bits. Moreover, by using compressed bitmaps 
\cite{RRR07}, the space is reduced to $nH_0(S)+o(n(H_0(S)+1))$ bits,
albeit in practice the navigation is slowed down.

The algorithm to compute a canonical Huffman code \cite{SK64} starts from the
code length assignments $\ell_a$ produced by the standard Huffman algorithm,
and produces a particular Huffman tree with the same code lengths. First, it
computes $\ell_{min}$ and $\ell_{max}$, the minimum and maximum code lengths,
and array $nCodes[\ell_{min},\ell_{max}]$, where $nCodes[\ell]$ is the number 
of codes of length $\ell$. Then, the algorithm assigns the codes as follows:
\begin{enumerate}
\item $fst[\ell_{min}] = 0^{\ell_{min}}$ (i.e., $\ell_{min}$ 0s) is the first
code of length $\ell_{min}$.
\item All the codes of a given length $\ell$ are consecutive numbers, 
from $fst[\ell]$ to $last[\ell] = fst[\ell]+nCodes[\ell]-1$.
\item The first code of the next length $\ell'>\ell$ that has
$nCodes[\ell']>0$ is $fst[\ell']=2^{\ell'-\ell}(last[\ell]+1)$.
\end{enumerate}

Note that rule 2 ensures that all codes of a given level are consecutive 
numbers and the first of their length, whereas rule 3 guarantees that the set 
of produced codes is prefix-free. By interpreting the bit 0 as the right child 
and the bit 1 as the left child, we have that all the leaves at any level are 
the rightmost nodes. Figure~\ref{fig:canonical} illustrates the standard and 
the levelwise deployment of a canonical Huffman code.

\begin{figure}[t]
%
  \hspace*{-2mm}\includegraphics[height=4.6cm]{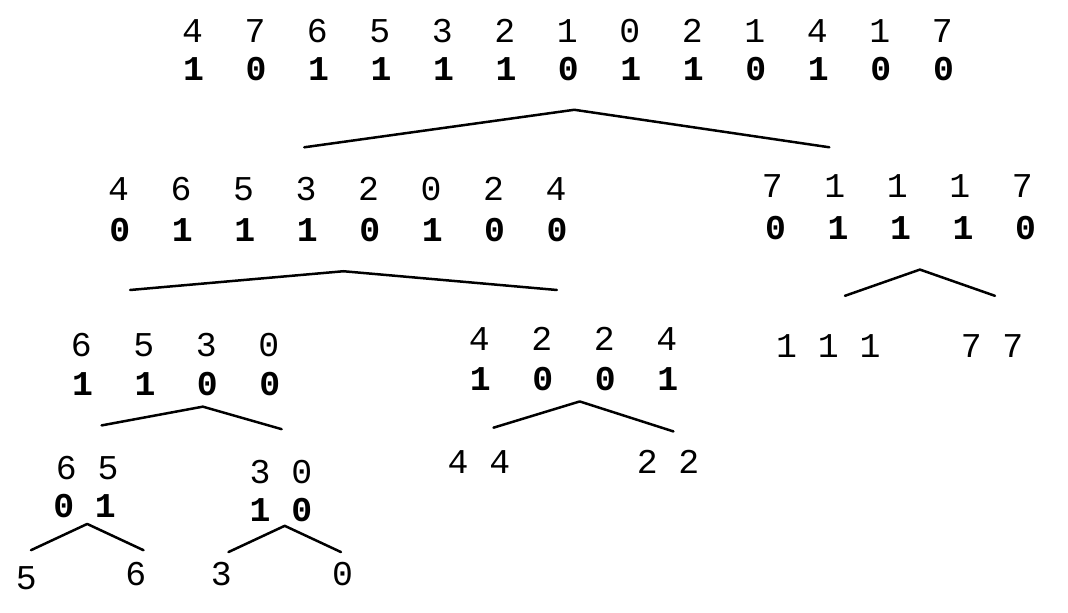}   

\vspace*{-4.6cm}

 \hfill \includegraphics[height=4.6cm]{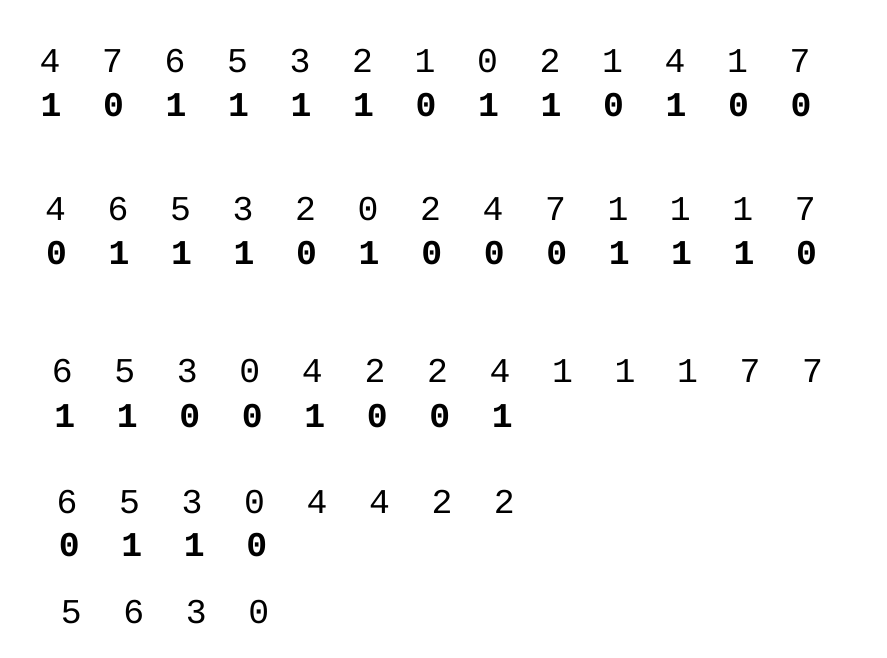} \hspace*{-2mm}


\caption{On the left, the pointer-based canonical Huffman code for our running
example. On the right, its levelwise representation. Note that from now on we 
interpret the bit 0 as going right and the bit 1 as going left.}
\label{fig:canonical}
\end{figure}

\section{The Wavelet Matrix}
\label{sec:wmatrix}

The idea of the wavelet matrix is to break the assumption that the children
of a node $v$, at interval $\B_\ell[s_v,e_v]$, must be aligned to it and
occupy the interval $\B_{\ell+1}[s_v,e_v]$. Freeing the structure from this
assumption allows us to design a much simpler mapping mechanism
from one level to the next: {\em all} the zeros of the level go left, and
{\em all} the ones go right. For each level, we will store a single integer
$z_\ell$ that tells the number of 0s in level $\ell$. This requires just
$O(\lg n \lg\sigma)$ bits, which is insignificant, and allows us to implement
the pointerless mechanisms in a simpler and faster way.

More precisely, if $\B_\ell[i]=0$, then the corresponding position at level
$\ell+1$ will be $\rank_0(\B_\ell,i)$. If $\B_\ell[i]=1$, the position at
level $\ell+1$ will be $z_\ell + \rank_1(\B_\ell,i)$. Note that we can map the
position without knowledge of the boundaries of the node the position belongs.
Still, every node $v$ at level $\ell$ occupies a contiguous range in
$\B_\ell$, as proved next. 

\begin{lem}
All the bits in any bitmap $\B_\ell'$ of the pointerless wavelet tree that 
correspond to a wavelet tree node $v$ are also contiguous in the bitmap
$\B_\ell$ of the the wavelet matrix.
\end{lem}
\begin{proof}
This is obviously true for the root $v=\nu$, as it corresponds to the whole
$\B_0' = \B_0$. Now, assuming it is true for a node $v$, with interval 
$\B_\ell[s_v,e_v]$, all the positions with $\B_\ell[i]=0$ for 
$s_v \le i \le e_v$ will be mapped to consecutive positions
$\B_{\ell+1}[\rank_0(\B_\ell,i)]$, and similarly with positions 
$\B_\ell[i]=1$.
\end{proof}

Figure~\ref{fig:wmatrix} illustrates the wavelet matrix, where it can
be seen that the blocks of the wavelet tree are maintained, albeit in
different order. We now describe how to carry out the operations under the
strict and the extended variants.

\begin{figure}[t]
  \includegraphics[height=3.8cm]{wtnp.pdf}   \hfill

\vspace*{-3.8cm}

 \hfill \includegraphics[height=3.65cm]{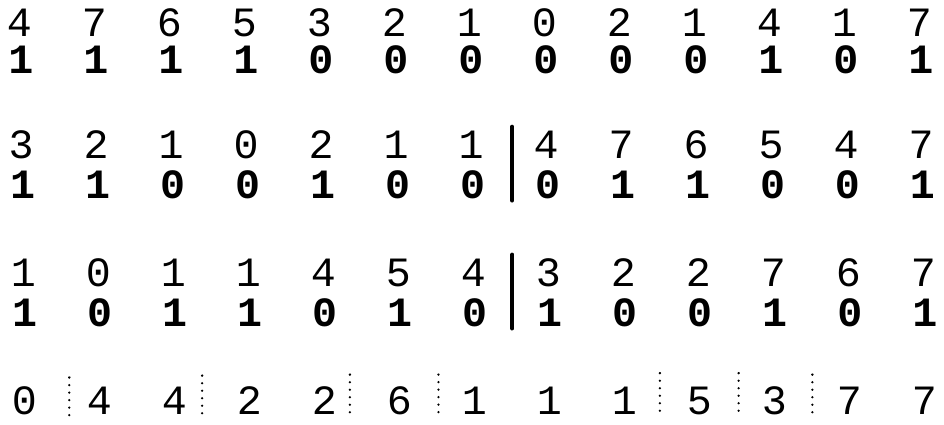}

\caption{On the left, the pointerless wavelet tree of Figure~\ref{fig:wtree}.
On the right, the wavelet matrix over the same sequence. One vertical 
line per level represents the position stored in the $z_\ell$ values.}
\label{fig:wmatrix}
\end{figure}

\paragraph{The strict variant.}
To carry out $\access(S,i)$, we first set $i_0$ to $i$. Then,
if $\B_0[i_0]=0$, we set $i_1$ to $\rank_0(\B_0,i_0)$. Else we set $i_1$ to 
$z_0 + \rank_1(\B_0,i_0)$.
Now we descend to level 1, and continue until reaching a leaf. The sequence
of bits $\B_\ell[i_\ell]$ read along the way form the value $S[i]$ (or, said 
another way, we maintain the interval $[\alpha_v,\omega_v)$ and upon reaching 
the leaf it holds $S[i]=\alpha_v$). Note that we have carried out only one 
binary $\rank$ operation per level, just as the standard wavelet tree.

Consider now the computation of $\rank_a(S,i)$. This time we need to keep
track of the position $i$, and also of the position preceding the range,
initially $p_0=0$. At each node $v$ of depth $\ell$, if
$a < 2^{\lceil \lg(\omega_v-\alpha_v)\rceil-1}$, then we go ``left'' by
mapping $p_{\ell+1}$ to $\rank_0(\B_\ell,p_\ell)$ and $i_{\ell+1}$ to 
$\rank_0(\B_\ell,i_\ell)$. Otherwise,
we go ``right'' by mapping $p_{\ell+1}$ to $z_\ell + \rank_1(\B_\ell,p_\ell)$ 
and $i_{\ell+1}$ to $z_\ell + \rank_1(\B_\ell,i_\ell)$. 
When we arrive at the leaf level, the
answer is $i_\ell-p_\ell$. Note that we have needed one extra binary $\rank$ 
operation
per original $\rank$ operation of the standard wavelet tree, instead of the
two extra operations required by the (strict) pointerless variant.

Finally, consider operation $\select_a(S,j)$. We first descend towards the
leaf of $a$ just as done for $\rank_a(S,i)$, keeping track only of $p_\ell$. 
When we arrive at the last level, $p_\ell$ precedes the range corresponding to 
the leaf of $a$, and thus we wish to track upwards position $j_\ell=p_\ell+j$. 
The upward tracking of a position $\B_{\ell}[j_\ell]$ is simple: If we went 
left from level $\ell-1$, then this position was mapped from a 0 in 
$\B_{\ell-1}$, and therefore it came from
$j_{\ell-1} = \B_{\ell-1}[\select_0(\B_\ell,j_\ell)]$. 
Otherwise, position $j_\ell$ was mapped from a 1, and thus it came from 
$j_{\ell-1} = \B_{\ell-1}[\select_1(\B_\ell,j_\ell-z_\ell)]$.
When we arrive at the root bitmap, $j_0$ is the answer. Note that we have
needed one extra binary $\rank$ per original binary $\select$ required by
the standard wavelet tree. We remind that in practice $\rank$ is much less
demanding, so this overhead is low.
Algorithm~\ref{alg:matrix} gives the pseudocode.

\begin{algorithm}[t]
\caption{Wavelet matrix algorithms (strict variant): 
On the wavelet matrix of sequence $S$,
$\mathbf{acc}(0,i)$ returns $S[i]$;
$\mathbf{rnk}(0,a,i,0)$ returns $\rank_a(S,i)$; and
$\mathbf{sel}(0,a,j,0)$ returns $\select_a(S,j)$.
For simplicity we have omitted the computation of $[\alpha_v,\omega_v)$.}
\label{alg:matrix}
\begin{tabular}{ccc}
\begin{minipage}{0.29\textwidth}
$\mathbf{acc}(\ell,i)$

\vspace{-0.1cm}
\begin{algorithmic}
\IF{$\omega_v-\alpha_v=1$}
   \RET $\alpha_v$
\ENDIF
\IF{$\B_\ell[i]=0$}
   \STATE $i \leftarrow \rank_0(\B_\ell,i)$
\ELSE
   \STATE $i \leftarrow \rank_1(\B_\ell,i)$
\ENDIF
\RET $\mathbf{acc}(\ell{+}1,i)$
\STATE
\STATE
\STATE
\end{algorithmic}
\end{minipage}
&
\begin{minipage}{0.33\textwidth}
$\mathbf{rnk}(\ell,a,i,p)$

\vspace{-0.1cm}
\begin{algorithmic}
\IF{$\omega_v-\alpha_v=1$}
   \RET $i-p$
\ENDIF
\IF{$a < 2^{\lceil \lg(\omega_v-\alpha_v)\rceil-1}$}
   \STATE $p \leftarrow \rank_0(\B_\ell,p)$
   \STATE $i \leftarrow \rank_0(\B_\ell,i)$
\ELSE
   \STATE $p \leftarrow z_\ell+\rank_1(\B_\ell,p)$
   \STATE $i \leftarrow z_\ell+\rank_1(\B_\ell,i)$
\ENDIF
\RET $\mathbf{rnk}(\ell{+}1,a,i,p)$
\STATE
\end{algorithmic}
\end{minipage}
&
\begin{minipage}{0.36\textwidth}
$\mathbf{sel}(\ell,a,j,p)$

\vspace{-0.1cm}
\begin{algorithmic}
\IF{$\omega_v-\alpha_v=1$}
   \RET $p+j$
\ENDIF
\IF{$a < 2^{\lceil \lg(\omega_v-\alpha_v)\rceil-1}$}
   \STATE $p \leftarrow \rank_0(\B_\ell,p)$
   \STATE $j \leftarrow \mathbf{sel}(\ell{+}1,a,j,p)$
   \RET $\select_0(\B_\ell,j)$
\ELSE
   \STATE $p \leftarrow z_\ell + \rank_1(\B_\ell,p)$
   \STATE $j \leftarrow \mathbf{sel}(\ell{+}1,a,j,p)$
   \RET $\select_1(\B_\ell,j{-}z_\ell)$
\ENDIF
\end{algorithmic}
\end{minipage}
\end{tabular}
\end{algorithm}

\paragraph{The extended variant.}
We can speed up $\rank$ and $\select$ operations if the array $C$ that points
to the starting positions of each symbol in the last level bitmap is
available. First, we note that for $\rank_a(S,i)$ we do not need anymore to
keep track of $p_\ell$, since all we need at the end is to return $i_\ell-
C[a]$. Thus the cost becomes similar to that of the standard wavelet tree,
which was not achieved with the extended variant of the pointerless wavelet
tree.

For $\select_a(S,j)$ we can avoid the first downward traversal, as in the
pointerless wavelet tree, and use the same technique to determine whether
we came from the left or from the right in the parent bitmap. Once again,
the cost becomes the same as in a standard wavelet tree, with no extra
$\rank$ operations required.
Algorithm~\ref{alg:matrixC} gives the detailed algorithm.

\begin{algorithm}[t]
\caption{Wavelet matrix algorithms (extended variant): 
On the wavelet matrix of sequence $S$,
$\mathbf{acc}(0,i)$ returns $S[i]$;
$\mathbf{rnk}(0,a,i)$ returns $\rank_a(S,i)$; and
$\mathbf{sel}(a,j)$ returns $\select_a(S,j)$.
For simplicity we have omitted the computation of $[\alpha_v,\omega_v)$
For simplicity we have omitted the computation of $[\alpha_v,\omega_v)$,
and in $\mathbf{sel}(a,j)$ we assume $C[a]$ refers
to level $\ell = \lceil\lg\sigma\rceil$, where in fact it could refer to
level $\ell = \lceil\lg\sigma\rceil-1$.}
\label{alg:matrixC}
\begin{tabular}{ccc}
\begin{minipage}{0.29\textwidth}
$\mathbf{acc}(\ell,i)$

\begin{algorithmic}
\IF{$\omega_v-\alpha_v=1$}
   \RET $\alpha_v$
\ENDIF
\IF{$\B_\ell[i]=0$}
   \STATE $i \leftarrow \rank_0(\B_\ell,i)$
\ELSE
   \STATE $i \leftarrow \rank_1(\B_\ell,i)$
\ENDIF
\RET $\mathbf{acc}(\ell{+}1,i)$
\end{algorithmic}
\ \\
\end{minipage}
&
\begin{minipage}{0.33\textwidth}
$\mathbf{rnk}(\ell,a,i)$

\begin{algorithmic}
\IF{$\omega_v-\alpha_v=1$}
   \RET $i-C[a]$
\ENDIF
\IF{$a < 2^{\lceil \lg(\omega_v-\alpha_v)\rceil-1}$}
   \STATE $i \leftarrow \rank_0(\B_\ell,i)$
\ELSE
   \STATE $i \leftarrow z_\ell+\rank_1(\B_\ell,i)$
\ENDIF
\RET $\mathbf{rnk}(\ell{+}1,a,i)$
\end{algorithmic}
\ \\
\end{minipage}
&
\begin{minipage}{0.36\textwidth}
$\mathbf{sel}(a,j)$
\begin{algorithmic}
\STATE $\ell \leftarrow \lceil \lg \sigma \rceil$,
       $d \leftarrow 1$
\STATE $j \leftarrow C[a]+j$
\WHILE{$\ell \ge 0$}
   \IF{$a ~\mathrm{mod}~ 2^d = 0$}
      \STATE $j \leftarrow \select_0(\B_\ell,j)$
   \ELSE
      \STATE $j \leftarrow \select_1(\B_\ell,j{-}z_\ell)$
   \ENDIF
   \STATE $\ell \leftarrow \ell-1$,
          $d \leftarrow d+1$
\ENDWHILE
\RET $j$
\end{algorithmic}

\end{minipage}
\end{tabular}
\end{algorithm}

\paragraph{Range searches.}
Range searches for rectangles $[x_1,x_2] \times [y_1,y_2]$ require essentially 
that we are able to track the points $x_1$ and $x_2$ downwards in the tree.
Thus the same wavelet matrix mechanism for $\rank$ can be used. Since we are 
only interested
in the value $x_2-x_1$ at the traversed nodes, we do not need to keep track
of $p$, even in the strict variant (the extended variant requires too much
space in this scenario). As a result, we need the same number of $\rank$
operations as in a pointer-based representation, and get rid of the two
extra $\rank$ operations required by the pointerless wavelet tree.
Algorithm~\ref{alg:rangematrix} gives the pseudocode.

\begin{algorithm}[t]
\caption{Range search algorithms on the wavelet matrix: 
$\mathbf{count}(0,x_1,x_2,y_1,y_2)$
returns $\rc(P,x_1,x_2,y_1,y_2)$ on the wavelet tree of sequence $P$; and
$\mathbf{report}(0,x_1,x_2,y_1,y_2)$ outputs all those $y$, where a point
with coordinate $y_1 \le y \le y_2$ appears in $P[x_1,x_2]$. For
simplicity we have omitted the computation of $[\alpha_v,\omega_v)$.}
\label{alg:rangematrix}
\begin{tabular}{cc}
\begin{minipage}{0.5\textwidth}
$\mathbf{count}(\ell,x_1,x_2,y_1,y_2)$
\begin{algorithmic}
\IF{$x_1 > x_2 ~\lor~ [\alpha_v,\omega_v] \cap [y_1,y_2] = \emptyset$}
        \RET 0
\ELSIF{$[\alpha_v,\omega_v] \subseteq [y_1,y_2]$}
        \RET $x_2-x_1+1$
\ELSE
        \STATE $x_1^l \leftarrow \rank_0(\B_\ell,x_1-1)+1$
        \STATE $x_2^l \leftarrow \rank_0(\B_\ell,x_2)$
        \STATE $x_1^r \leftarrow x_1-x_1^l+1$, $x_2^r \leftarrow x_2-x_2^l$
        \RET $\mathbf{count}(\ell{+}1,x_1^l,x_2^l,y_1,y_2)$
        \STATE \hspace{0.97cm} $+ \mathbf{count}(\ell{+}1,x_1^r,x_2^r,y_1,y_2)$
\ENDIF
\end{algorithmic}
\end{minipage}
&
\begin{minipage}{0.5\textwidth}
$\mathbf{report}(v,x_1,x_2,y_1,y_2)$
\begin{algorithmic}
\IF{$x_1 > x_2 ~\lor~ [\alpha_v,\omega_v] \cap [y_1,y_2] = \emptyset$}
        \RET
\ELSIF{$\omega_v-\alpha_v=1$}
        \OUTPUT $\alpha_v$
\ELSE
        \STATE $x_1^l \leftarrow \rank_0(\B_\ell,x_1-1)+1$
        \STATE $x_2^l \leftarrow \rank_0(\B_\ell,x_2)$
        \STATE $x_1^r \leftarrow x_1-x_1^l+1$, $x_2^r \leftarrow x_2-x_2^l$
        \STATE $\mathbf{report}(\ell{+}1,x_1^l,x_2^l,y_1,y_2)$
        \STATE $\mathbf{report}(\ell{+}1,x_1^r,x_2^r,y_1,y_2)$
\ENDIF
\end{algorithmic}
\end{minipage}
\end{tabular}
\end{algorithm}

\paragraph{Construction.}
Construction of the wavelet matrix is even simpler than that of the pointerless
wavelet tree, because we do not need to care about node boundaries. At the
first level we keep in bitmap $\B_0$ the highest bits of the symbols in $S$,
and then stably sort $S$ by those highest bits. Now we keep in bitmap $\B_1$
the next-to-highest bits, and stably sort $S$ by those next-to-highest bits.
We continue until considering the lowest bit. This takes $O(n\lg\sigma)$ time.

Indeed, we can build the wavelet matrix almost in place,
by removing the highest bits after using them and packing the symbols
of $S$. This frees $n$ bits, where we can store the bitmap $\B_0$ we have just
generated, and keep doing the same for the next levels. We generate the
$o(n\lg\sigma)$-space indexes at the end. Thus the construction space is
$n\lceil\lg\sigma\rceil + \max(n,o(n\lg\sigma))$ bits. Other more
sophisticated techniques \cite{CNS11,Tis11} may use even less space.

\section{The Compressed Wavelet Matrix}
\label{sec:comprwmatrix}

Just as on the pointerless wavelet tree, we can achieve zero-order entropy 
with the wavelet matrix by replacing the plain representations 
of bitmaps $\B_\ell$ by compressed ones \cite{RRR07}, the space becoming
$nH_0(S) + o(n\lg\sigma)$ bits. 
Compared to obtaining zero-order entropy using Huffman shaped trees, this
solution has several disadvantages, as explained: (1) the compressed bitmaps 
are slower to operate than in a plain representation; (2) the number of 
operations on a Huffman shaped tree is lower on average than on a balanced 
tree; (3) the Huffman shaped wavelet tree is more compact, as it reduces the 
redundancy from $o(n\lg\sigma)$ to $o(n(H_0(S)+1))$ (albeit a small 
$O(\sigma\lg n)$-bit space term is added to hold the Huffman model); 
(4) the bitmap compression
can be additionally combined with the Huffman shape, obtaining further 
compression (yet higher time).

The idea is the same as in Section~\ref{sec:levelhuff}: Arrange the codes so
that all the leaves are grouped to the right of the bitmaps $\B_\ell$.
However, because of the reordering of nodes produced by the wavelet matrix,
the use of canonical Huffman codes does not guarantee that the leaves of the
same level are contiguous. In the wavelet matrix, the position of a code $c$ in $\B_{\ell+1}$ depends only on the position of $c$ in $\B_{\ell}$ and on the 
bit of $c$ in that level, $c[\ell]$. Figure~\ref{fig:example} illustrates an
example of a canonical set of codes where the first 16 shortest codewords take 
values from $00000$ to $01111$ and the remaining 32 from $100000$ to $111110$. 
The figure shows the relative positions of the codes at successive levels of 
the wavelet matrix for a sequence $\dots c_8,c_{12},c_{32},c_{48}\dots$, where 
$c_8=01000$, $c_{12}=01100$, $c_{32}=100000$, and $c_{48}=110000$. 
As we can see, codes $c_8$ and $c_{12}$ finish at level $5$ but they are not 
contiguous since there is a $c_{48}$ between them. 

\begin{figure}[t]

\begin{eqnarray*}
\ell=1 & \dots,c_8,c_{12},c_{32},c_{48}\dots \\
\ell=2 & \dots,c_8,c_{12},\dots|\dots,c_{32},c_{48}\dots \\
\ell=3 & \dots,c_{32}\dots|\dots,c_8,c_{12},\dots,c_{48}\dots \\
\ell=4 & \dots,c_{32},\dots,c_8,\dots,c_{48}|\dots,c_{12},\dots \\
\ell=5 & \dots,c_{16},\dots,c_8,\dots,c_{48},\dots,c_{12}\dots | \dots \\
\ell=6 & \dots,c_{32},c_{48},\dots|\dots 
\end{eqnarray*}

\caption{Example of a sequence of canonical codes along wavelet matrix levels,
showing that the leaves do not span a contiguous area. The the vertical bar 
$``|"$ marks the points $z_{\ell}$.}
\label{fig:example}
\end{figure}

We require a distinct mechanism to design an optimal prefix-free code that 
guarantees that, under the shuffling rules of the wavelet matrix, all the 
leaves at any level form a contiguous area to the right of the bitmap.

We start by studying how the wavelet matrix sorts the codes at each level.
Consider a pair of codes $c_1[1,\ell_1]$ and $c_2[1,\ell_2]$. Depending on 
their bits at a given level $\ell$ of the wavelet matrix, two cases are 
possible: $(a)$ $c_1[\ell]=c_2[\ell]$ and then the relative positions of $c_1$ 
and $c_2$ stay the same at level $\ell+1$, or $(b)$ $c_1[\ell]\neq c_2[\ell]$ 
and then their relative positions in level $\ell+1$ depend on the relation
between $c_1[\ell]$ and $c_2[\ell]$.
This yields the following lemma: 

\begin{lem}
\label{lemma:order.prefix}
 In a wavelet matrix, given any pair of codes $c_1$ and $c_2$, $c_1$ appears
before(after) $c_2$ in $\B_\ell$ if, for some $0 \le i < \ell$, it holds
$c_1[\ell-i,\ell-1]=c_2[\ell-i,\ell-1]$ and $c_1[\ell-i-1]=0(1)\neq c_2[\ell-i-1]$.
\end{lem}

\begin{proof}
If $c_1[\ell-i,\ell-1]=c_2[\ell-i,\ell-1]$, then $c_1$ and $c_2$ transitively 
keep their relative positions from level $\ell-i$ to level $\ell$. Instead,
$c_1[\ell-i-1]\neq c_2[\ell-i-1]$ makes their ordering in level $\ell-i$
dependent only on how $c_1[\ell-i-1]$ and $c_2[\ell-i-1]$ compare to each
other.
\end{proof}

As a second step, assume we want to design a set of fixed-length codes 
$\{ c_a,\, a \in [0,\sigma) \}$ such that $c_a<c_b$ iff the area of $c_a$ is 
before that of $c_b$ in $\B_{\lceil \lg \sigma \rceil}$. That is, we want the
codes to be listed in order in the last level.
Let $inv:\{0,1\}^{\mathbb{N}^+} \times \mathbb{N}^+ \rightarrow 
\{0,1\}^{\mathbb{N}^+}$ be defined as
$inv(c[1,\ell],\ell) = c^{-1}[1,\ell]$, where $c^{-1}[i]=c[\ell-i+1]$ for 
all $1 \le i \le \ell$. That is, $inv(c,\ell)$ takes number $c$ as a codeword
of $\ell$ bits and returns the code obtained by reading $c$ backwards.
Then, the following lemma holds: 

\begin{lem}
\label{lemma:inverted.codes}
Given any two values $i,j \in [0,\sigma)$ where $i<j$, code 
$inv(i,\lceil \lg \sigma \rceil)$ is located to the left of code 
$inv(j,\lceil \lg \sigma \rceil)$ in the bitmap 
$\B_{\lceil \lg \sigma \rceil}$ of a wavelet matrix that 
uses such codes.
\end{lem}

\begin{proof}
Let $\tau_i=inv(i,\lceil \lg \sigma \rceil)$ and $\tau_j=inv(j,\lceil \lg \sigma \rceil)$. 
If $\tau_i$ and $\tau_j$ do not share any common suffix, then their relative 
positions in $\B_{\lceil \lg \sigma \rceil}$ depend only on their last bit and 
the relation is given by that bit. Otherwise, $\tau_i$ and $\tau_j$ share a 
common suffix of length $\lceil\lg\sigma\rceil-\delta+1 
\in[1,\lceil \lg \sigma \rceil]$, that is, 
$\tau_i[\delta,\lceil\lg\sigma\rceil]= \tau_j[\delta,\lceil\lg\sigma\rceil]$.
Then, according to Lemma~\ref{lemma:order.prefix},
$\tau_i$ is before $\tau_j$ iff $\tau_i[\delta]< \tau_j[\delta]$. 
In both cases the relation is given by the last distinct bit of the codes,
or the first if they are read backwards. Since the codes are of the same 
length, comparing by the first distinct bit is equivalent to comparing
numerically. That is, $\tau_i$ is before $\tau_j$ iff
$inv(\tau_i,\lceil \lg \sigma \rceil)<
inv(\tau_j,\lceil \lg \sigma \rceil)$. In turn, since
$inv(inv(c,\ell),\ell)=c$, this is equivalent to $i<j$.
\end{proof}

The lemma gives a way to force a desired order in a set of fixed-length codes:
Given symbols $a \in [0,\sigma)$, we can assign them codes $c_a = inv(a,\lceil
\lg \sigma \rceil)$ to ensure that the areas become ordered in 
$\B_{\lceil \lg \sigma \rceil}$. As a side note, we observe that we could have
retained the symbol order natively in the wavelet matrix if we had chosen to
decompose the symbols from their least to their most significant bit, and not
the other way (in this case the wavelet matrix is actually radix-sorting the
values). This brings problems in the extended variants, however, 
because the resulting range of codes has unused entries if $\sigma$ is not
a power of 2. For example, consider alphabet $0,1,2,3,4 = 000,\ldots,100$; after
reversing the bits we obtain numbers $0,1,2,4,6$, so we need to allocate 7 
cells for $C$ instead of 5. The size of $C$ can double in the 
worst case. We cannot either directly use the idea of reversing the canonical
Huffman codes, because the codes could not be prefix-free anymore. A more
sophisticated scheme, based on Lemma~\ref{lemma:inverted.codes}, is required.

Assume we have obtained the desired code lengths $\ell_a$, as well as the 
array $nCodes$ from the canonical Huffman construction.
We generate the final Huffman tree in levelwise order.
The simplest description is as follows. We start with a set
of valid codes $\mathcal{C} = \{ 0, 1 \}$ and level $\ell=1$. At each level
$\ell$, we remove from $\mathcal{C}$ the $nCodes[\ell]$ codes $c$ with minimum
$inv(c,\ell)$ value. The removed nodes are assigned to the $nCodes[\ell]$
symbols that require codes of length $\ell$. Now we replace each code $c$ 
remaining in $\mathcal{C}$, by two new codes, $c:0$ and $c:1$, and continue
with level $\ell+1$. It is clear that this procedure generates a prefix-free
set of codes that, when reversed, satisfy that the codes finishing at a level
are smaller than those that continue.

It is not hard to see that the total cost of this algorithm is linear.
There are two kind of codes inserted in $\mathcal{C}$:
those that will be chosen for a code and those that will not. There are
exactly $\sigma$ nodes of the first class, whereas for each node of the second
class we insert other two codes in $\mathcal{C}$. Therefore the total number
of codes ever inserted in $\mathcal{C}$ adds up to $O(\sigma)$. The codes to
use at each level $\ell$ can be obtained by linear-time selection over the
set of codes just extended (sorting codes by $inv(c,\ell)$), thus adding up to 
$O(\sigma)$ time as well.

Figure~\ref{fig:huffmatrix} gives an example of the construction.

\begin{figure}[t]

\vspace*{1cm}

 \includegraphics[height=2.8cm]{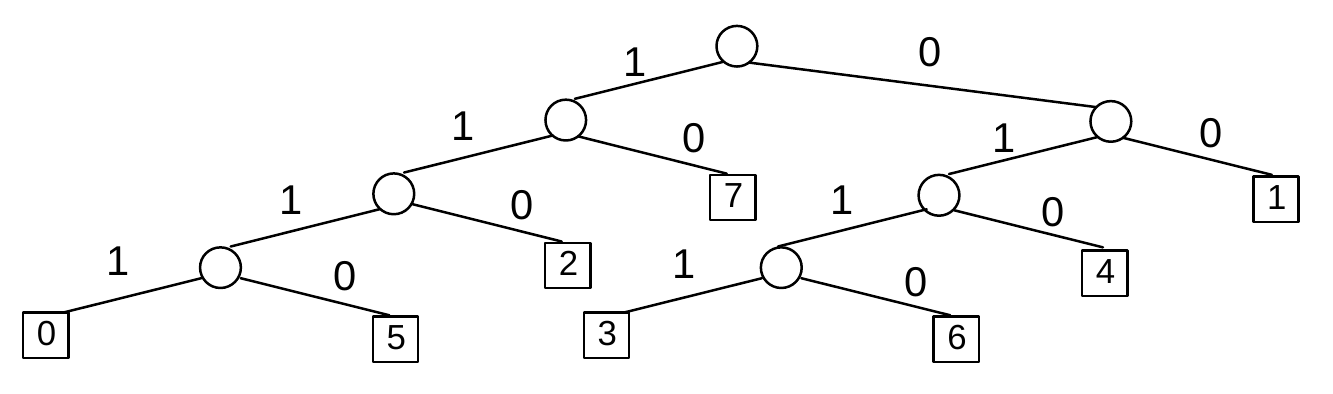}   \hfill
\vspace*{-3.75cm}

 \hfill \includegraphics[height=4.8cm]{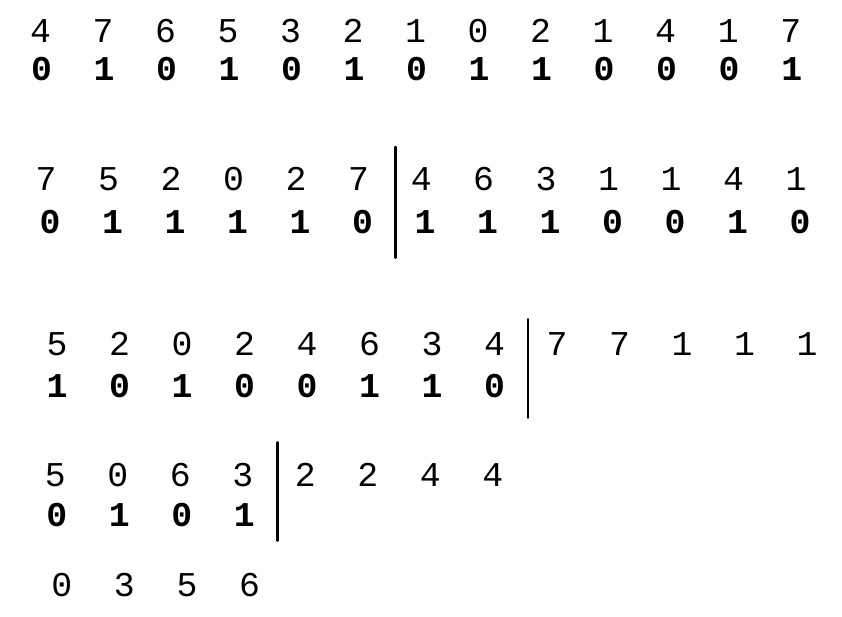}

  \caption{On the left, the Huffman tree resulting from our code reassignment
algorithm on the running example. On the right, the resulting Huffman shaped 
wavelet matrix.}
  \label{fig:huffmatrix}
\end{figure}

\section{Experimental Results}

Our implementations build over the wavelet tree implementations of \libcds, 
a library implementing several space-efficient data structures.%
\footnote{{\tt https://github.com/fclaude/libcds}}
For each wavelet tree/matrix
variant we present two versions, \verb|CM| and \verb|RRR|. The first one
corresponds to using the the $\rank/\select$ enabled bitmap 
implementation~\cite{GGMN05} of the proposals of Clark~\cite{Cla96} and 
Munro~\cite{Mun96}, choosing 5\% space overhead over the plain bitmap.
The second version, \verb|RRR|, corresponds to using the bitmap 
implementation~\cite{CN08} of the compressed bitmaps of Raman, Raman and 
Rao~\cite{RRR07}. 
The variants compared are the following:

\begin{itemize}
\item 
\verb|WT|: standard pointer-based wavelet tree;
\item 
\verb|WTNP|: the (extended) pointerless wavelet tree (``No Pointers'');
\item 
\verb|WM|: the (extended) wavelet matrix (Section~\ref{sec:wmatrix});
\item
\verb|HWT|: the Huffman shaped standard pointer-based wavelet tree;
\item
\verb|HWTNP|: the Huffman shaped extended levelwise wavelet tree
	(Section~\ref{sec:levelhuff});
\item
\verb|HWM|: the Huffman shaped (extended) wavelet matrix
	(Section~\ref{sec:comprwmatrix});
\item	
\verb|AP|: the alphabet-partitioned data structure of Barbay et 
al.~\cite{BCGNN13}, which is the best state-of-the-art alternative to wavelet
trees.
\end{itemize}

These names are composed with the bitmap implementations by appending the
bitmap representation name. For example, we call \verb|WT-RRR| the standard
pointer-based wavelet tree with all bitmaps represented with Raman, Raman and
Rao's compressed bitmaps. \verb|AP| uses always \verb|CM| bitmaps,
which is the best choice for this structure.

Note that all the pointerless structures use the array $C$. The extended
versions generally achieve space very close to the strict ones and perform
much faster.

\subsection{Datasets}

In order to evaluate the performance of $\access$, $\rank$ and $\select$, we 
use four different datasets:\footnote{Left at {\tt http://lbd.udc.es/research/ECWTLA}}

\begin{itemize}

\item \verb|ESWiki|: Sequence of word identifiers generated by
stemming the Spanish Wikipedia\footnote{{\tt http://es.wikipedia.org}
dated 03/02/2010.} with the Snowball algorithm.
The sequence has length $n=200{,}000{,}000$, alphabet 
size $\sigma=1{,}634{,}145$, and zero-order entropy $H_0=11.12$. 
This sequence can be used to
simulate a positional inverted index \cite{CN08,AGO10,BFLN12}.

\item \verb|BWT|: The Burrows-Wheeler transform (BWT) \cite{BW94}
of \verb|ESWiki|.
The length and size of the alphabet, as well as the zero-order entropy,
match those of \verb|ESWiki|. However, \verb|BWT| has a much lower
high-order entropy \cite{Man01}. Many full-text compressed self-indexes 
\cite{FM05,FMMN07,NM07} use the BWT of the text they represent.

\item \verb|Indochina|: The concatenation of all adjacency lists
of Web graph {\tt Indochina\-2004}, available at the WebGraph project.%
\footnote{{\tt http://law.dsi.unimi.it}} The length of the
sequence is $n=100{,}000{,}000$, the alphabet size $\sigma=2{,}705{,}024$, 
and the entropy is $H_0=15.69$.
This representation
has been used to support forward and backward traversals on the graph 
\cite{CN08,CN10}.

\item \verb|INV|: Concatenation of inverted lists for a 
random sample of $2{,}961{,}510$ documents from the English Wikipedia.%
\footnote{{\tt http://en.wikipedia.org}} This sequence has length
$n=338{,}027{,}430$ and its alphabet size is $\sigma=2{,}961{,}510$. From this sequence
we extract the first $n=180{,}000{,}000$ elements with an alphabet of size $\sigma=1{,}590{,}398$ 
and an entropy of $H_0=19.01$.
This sequence has been used to simulate document inverted indexes 
\cite{NP10,GNP11}.

\end{itemize}

In order to evaluate the range search performance over discrete grids, we use 
the following three datasets formed by synthetic and real collections of
MBRs (Minimum Bounding Rectangles of objects).
We insert the two opposite corners of each MBR as points in our dataset.

\begin{itemize}
 \item \verb|Zipf|: A synthetic collection of 
 $1{,}000{,}000$ MBRs 
 with a Zipfian distribution (world size = $1{,}000\times 1{,}000$,
$\rho=1)$.%
\footnote{{\tt http://lbd.udc.es/research/serangequerying}}
 \item \verb|Gauss|: A synthetic collection of contains $1{,}000{,}000$ MBRs
 with a Gaussian distribution (world size = $1{,}000\times 1{,}000$,
$\mu=500,\sigma=200)$.%
\footnotemark[\value{footnote}]
 \item \verb|Tiger|: A real collection of $2{,}249{,}727$ MBRs from 
 California roads, available at the U.S. Census Bureau.%
\footnote{{\tt http://www.census.gov/geo/www/tiger}}
\end{itemize}

For range searches we cannot use Huffman compression, because the order
of the symbols is not maintained at the leaves. \verb|AP| also shuffles the
alphabet, so it cannot be used in this scenario. Extended variants are not
a good option either, because in this case it holds $\sigma=n$. Thus we test 
only the strict variants of $\verb|WTNP|$ and $\verb|WM|$.

\subsection{Measurements}


To measure performance we generated $100{,}000$ inputs for each query
and averaged their execution time, running each query $10$ times.
The $\access(S,i)$ queries were generated by choosing positions $i$ uniformly 
at random in $[1,n]$. Queries $\rank_a(S,i)$ were generated by choosing $i$
uniformly at random, and then setting $a=S[i]$. Each $\select_a(S,j)$ query 
was generated by first choosing a position $i$ at random in $[1,n]$, then 
setting $a=S[i]$, and finally choosing $j$ at random in $[1,\rank_a(S,n)]$.
The resulting distribution is the most common in applications, and it obtains
the $O(H_0(S)+1)$ average time performance in the Huffman shaped variants.

To measure the performance on point grids, for synthetic collections 
we generate sets of queries covering $0.001\%$, $0.01\%$,
$0.1\%$, and $1\%$ of the grid area. 
The sets contain $1{,}000$ queries, each with a
ratio between both axes varying uniformly at random between $0.25$ and $2.25$. 
For the real data set \verb|Tiger|, we use as queries the following four collections 
(also available for downloading at the Web site of \verb|Tiger|): \verb|Block|
(groups of buildings), \verb|BG| (block groups), \verb|SD| (elementary, secondary, and unified school districts), and \verb|COUSUB| (country 
subdivisions). 

The machine used is an Intel(R) Xeon(R) E5620 running at $2.40$GHz
with $96$GB of RAM memory. The operating system is
GNU/Linux, Ubuntu 10.04, with kernel 2.6.32-33-server.x86\_64. All our
implementations use a single thread and are coded in {\tt C++}. The
compiler is \verb|gcc| version $4.4.3$, with \verb|-O9| optimization.

\subsection{Results on Sequences}

Figures~\ref{chart:access} to \ref{chart:select} show the time and space for 
the different data structures and configurations for $\access$, $\rank$ and
$\select$ queries. The black vertical bar on the plots shows the value of $H_0$.
The bitmaps are parametrized by setting their sampling 
values to $32$, $64$, and $128$. In the case of \verb|AP|, these bitmap 
samplings are combined with permutation samplings 4, 16, and 64, respectively,
and all are run with $\ell_{min}=10$, as in previous work \cite{BCGNN13}.

\paragraph{Space.}
We start by discussing the space usage, which we measure in bits per symbol 
(bps). First we note that the \verb|WM| variants use always the same space as
the corresponding \verb|WTNP| variants (while being faster, as we discuss
soon). The space of \verb|WTNP-CM| and \verb|WM-CM| is obviously
close to $\lceil \lg\sigma\rceil$ bps. The extra space incurred by \verb|WT-CM|
is the overhead of the wavelet tree pointers, and is roughly proportional to 
$\sigma/n$ (times some implementation-dependent constant). This amounts to
nearly 4 bps in \verb|ESWiki| and \verb|BWT|, but 3.5 times more (14 
bps) in \verb|Indochina|, as expected from its larger alphabet size,
and again 4 bps in \verb|INV|. On the other hand, 
the space of \verb|HWTNP-CM| and \verb|HWM-CM| is always close to $H_0$ bits 
per symbol, plus a small extra to store the Huffman model.
The space overhead of
\verb|HWT-CM| on top of those corresponds, again, to the wavelet tree pointers.

The sampling parameter affects more sharply the \verb|RRR| variants, as they
store more data per sample.
The difference between \verb|WTNP-RRR| or \verb|WM-RRR| and \verb|WT-RRR| is
also proportional to $\sigma/n$, but this time the constant is higher because
the \verb|RRR| implementation needs more constants to be stored per bitmap
(i.e., per wavelet tree node). Thus the penalty is 6 bps on \verb|ESWiki| and
\verb|BWT|, 21 bps (3.5 times more) on \verb|Indochina|, and 7 bps 
on \verb|INV|. The same differences can be observed between
\verb|HWTNP-RRR| or \verb|HWM-RRR| and \verb|HWT-RRR|. We return later to the
fact that \verb|HWM-RRR| takes more space than \verb|HWTNP-RRR| on \verb|Indo| 
and \verb|INV|.

Finally, how \verb|WTNP-RRR|/\verb|WM-RRR| and \verb|HWTNP-RRR|/\verb|HWM-RRR|
compare to \verb|HWTNP-CM|/\verb|WM-CM| depends strongly on the type of 
sequence. In general, \verb|RRR| compression achieves the zero-order entropy
as an upper bound, but it can reach much less when the sequence has local
regularities. On the other hand, \verb|RRR| representation poses an additive 
overhead of 27\% of $\lg\sigma$, which corresponds to the $o(n\lg\sigma)$ 
overhead in this implementation \cite{CN08}. When combining Huffman and bitmap 
compression, this 27\% overhead acts over $H_0$ and not over $\lg\sigma$, which
brings it down, but on the other hand we must add the overhead of storing the
Huffman model. On \verb|ESWiki|, which has no special properties, the 27\%
overhead is around 5.7 bps, showing that \verb|RRR| compression reaches around 
8.3 bps, well below $H_0$. When combining with Huffman compression, this 
overhead becomes 14\%, that is, nearly 3 bps. Added to the 8.3 bps and to the 
1 bps of the Huffman model overhead, we still get slightly more space than 
plain Huffman compression, which is the best choice and reaches only 10\%
overhead over the zero-order entropy.

The picture changes when we consider \verb|BWT|. The Burrows-Wheeler transform
of \verb|ESWiki| boosts its higher-order compressibility \cite{Man01},
which is captured by \verb|RRR| compression \cite{MNimplicit07}, making 
\verb|RRR| compression reach the same space of Huffman compression, despite its
27\% space overhead. When combining both compressions, the result breaks the
zero-order entropy barrier by more than 10\% and becomes the best choice in
terms of space.

\verb|RRR| gives another surprising result on \verb|Indochina| and \verb|INV|,
where bitmap compression alone is more space-effective than in combination with
Huffman compression, and breaks the zero-order entropy by a large margin. This 
cannot be explained by high-order compressibility, as in this case the 
combination with Huffman would not harm. This behavior corresponds to the 
special nature of these sequences: the adjacency lists of the graph and the
inverted lists are sorted in increasing order. Long increasing sequences induce
long runs of 0s and 1s in the bitmaps of the wavelet trees and matrices. Those 
are retained in deeper levels when our partitioning by the most significant 
bit is used.\footnote{This is another advantage over using the least 
significant bit, as this partitioning breaks the runs faster.} The Huffman 
algorithm, instead, combines the nodes in unpredictable ways and destroys 
those long runs. 
Still, our Huffman algorithm maintains the order between those symbols whose
codewords have the same length, and thus the impact of this reordering is
not as high as it could be. Instead, the Huffman wavelet matrix completely
reshuffles the symbols. As a result, for example, the space of \verb|HWM-RRR|
exceeds that of \verb|HWTNP-RRR| by around 5 bps on \verb|Indo| and 6--7 bps on
\verb|INV|.

\paragraph{Time.}

The time results are rather consistent across collections. 
Let us first consider operation $\access$.
If we start considering the variants that do not use Huffman compression, 
we have that \verb|WT-RRR| is about 10\%--25\% slower than \verb|WT-CM|,
which is explained by a more complex implementation \cite{CN08}.
Instead, the pointerless variant, \verb|WTNP-CM|, is 20\%--25\% slower 
(recall that, in their extended variant, these require twice the number of
$\rank$ operations, but locality of reference makes them faster than twice 
the cost of one $\rank$ operation). However, \verb|WTNP-RRR| is about 40\%
slower than \verb|WT-RRR|, as the $\rank$ operation is slower and its higher
number impacts more on the total time (but still locality of reference makes
the percentage much less than 100\%). The wavelet matrix, instead, carries
out the same number of $\rank$ operations than the pointer-based wavelet tree,
so this time penalty disappears. Actually, \verb|WM-CM| is 8\%--14\% {\em
faster} than \verb|WT-CM|, and \verb|WM-RRR| is up to 4\% faster than
\verb|WT-RRR|. This may be due to less memory usage, which increases locality
of reference. Finally, the use of Huffman compression improves times by 
about $H_0/\lg\sigma$, as expected: times are reduced to about 50\%--60\% on
\verb|ESWiki| and \verb|BWT|, to about 65\%--85\% on 
\verb|Indochina|, and there is almost no reduction on \verb|INV|. 

The situation is basically the same for operation $\rank$, as expected from
the algorithms. The times are usually slightly lower because it is not 
necessary to access the bitmaps as we descend. 
The use of the wavelet matrix still gives essentially the same time (and even
slightly faster) than a pointer-based wavelet tree, and 
the use of Huffman shaped trees reduces the times by the same factors
as for $\access$, as expected.

The times of operation $\select$ show less difference between the standard and
the pointerless variants, because performing one extra $\rank$ operation is 
less relevant compared to the original (slower) $\select$ operation on the
bitmaps. One can see that \verb|WTNP-CM| is 30\%--40\% slower than
\verb|WT-CM| and that \verb|WTNP-RRR| is 35\%--50\% slower than \verb|WT-RRR|. 
The difference between plain and compressed bitmaps does not vary much, on
the other hand: \verb|WT-RRR| is 25\%--30\% slower than \verb|WT-CM|. What is
more surprising is that the wavelet matrix is clearly slower than the 
pointer-based wavelet trees: \verb|WM-CM| is 10\%--15\% slower than
\verb|WT-CM| and \verb|WM-RRR| is 20\%--30\% slower than \verb|WT-RRR|. The
reason is that the implementations of $\select$ \cite{GGMN05,CN08} proceed 
by binary search on the sampled values, thus their cost has in practice a
component that is logarithmic on the bitmap length. The bitmaps on the wavelet
tree nodes are shorter than $n$, whereas in the wavelet matrix (and the
pointerless wavelet tree) they are always of length $n$. Indeed, the wavelet
matrix is faster than the pointerless wavelet tree: \verb|WM-CM| is
20\%--25\% faster than \verb|WTNP-CM| and \verb|WM-RRR| is 12\%--15\% faster 
than \verb|WTNP-RRR|. Once again, the use of Huffman reduces all the times by
about the same space fraction obtained by zero-order compression.

\paragraph{Bottom line.}

On \verb|ESWiki|, where zero-order compression is the dominant space factor,
our Huffman shaped wavelet matrix, \verb|HWM-CM|, obtains the best space 
(only 10\% off the zero-order entropy) and the best time, by a good margin.

On \verb|BWT|, where higher-order compression is exploited by \verb|RRR|,
the space-time tradeoff map is dominated by the combination of \verb|HWM-RRR|
(minimum space) and \verb|HWM-CM| (minimum time), the two variants of our
Huffman shaped wavelet matrix. The former breaks the zero-order entropy 
barrier by about 10\%.

On \verb|Indochina| and \verb|INV|, where \verb|RRR| achieves space gains that 
are 
only degraded by Huffman compression, the dominant techniques are variants of 
the wavelet matrix: \verb|WM-RRR| (least space) and \verb|HWM-CM| (least time).
The former takes about 75\% of the zero-order entropy.

Summarizing, the wavelet matrix variants obtain the same space of the
pointerless wavelet trees, but they operate in about 65\% of their time,
reaching basically the same performance of the pointer-based variants but
much less space. As a result, they are always the dominant technique. Which
variant is the best, \verb|HWM-CM|, \verb|HWM-RRR| or \verb|WM-RRR|, depends
on the nature of the collection.

The comparison with \verb|AP| is interesting. In collections similar to
\verb|ESWiki|, Barbay et al.~\cite{BCGNN13} show that \verb|AP| generally 
achieves the best 
space and time among the alternatives \verb|WTNP-RRR|, \verb|WTNP-CM|, 
\verb|WT-CM|, and \verb|WT-RRR|, thus becoming an excellent choice in that
group. The new alternatives we have developed, however, clearly outperform 
\verb|AP|
in space: pointerless Huffman compression, and in particular Huffman wavelet 
matrices, improve upon the old wavelet tree alternatives in both space and 
time, using much less space than \verb|AP|. Still, \verb|AP| is a faster
representation, only slightly faster in operations $\access$ and
$\rank$, and definitely faster in operation $\select$.
The other collections also demonstrate that 
wavelet trees and matrices can exploit other compressibility features of the 
sequences apart from $H_0$, whereas \verb|AP| is blind to those (this is also
apparent in their experiments \cite{BCGNN13}, even using the basic wavelet
tree variants).

\subsection{Results on Point Grids}

Figures \ref{chart:range.count} and \ref{chart:range.report}
show the performance of \verb|WTNP| and \verb|WM| for $\rc$ and $\rr$ queries,
respectively. It turns out that, in the first level of each wavelet tree, the 
number of zeros and ones is highly unbalanced when the grid size is far from
the next power of 2. This makes the entropy of the first bitmap rather low,
whereas the other bitmaps are more balanced. On the other hand, the range
search algorithms spend just a few $\rank$ operations on the first bitmap.
To take advantage of this feature, we compress the bitmap of the first 
level of both data structures, \verb|WTNP| and \verb|WM|, with \verb|RRR|
and with a sampling of 32. The rest of bitmaps are represented using \verb|CM| 
with sampling rates of $32$, $64$, and $128$. 

In both figures \ref{chart:range.count} and \ref{chart:range.report} we 
append to the name of the data structure the name of the query set. This takes 
values in $\{Q0001, Q001, Q01, Q1\}$ in case of synthetic collections. In case 
of the real collection \verb|Tiger|, it takes values in 
$\{$BLock, BG, SD, COUSUB$\}$.

\begin{figure}[p]
\begin{center}
\includegraphics[angle=90,scale=1]{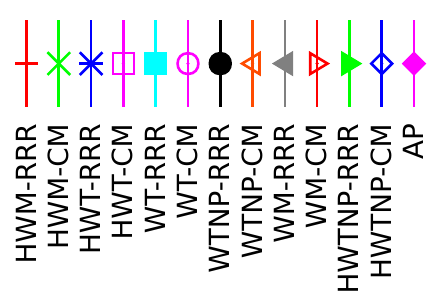} \\
\includegraphics[angle=90,scale=0.63]{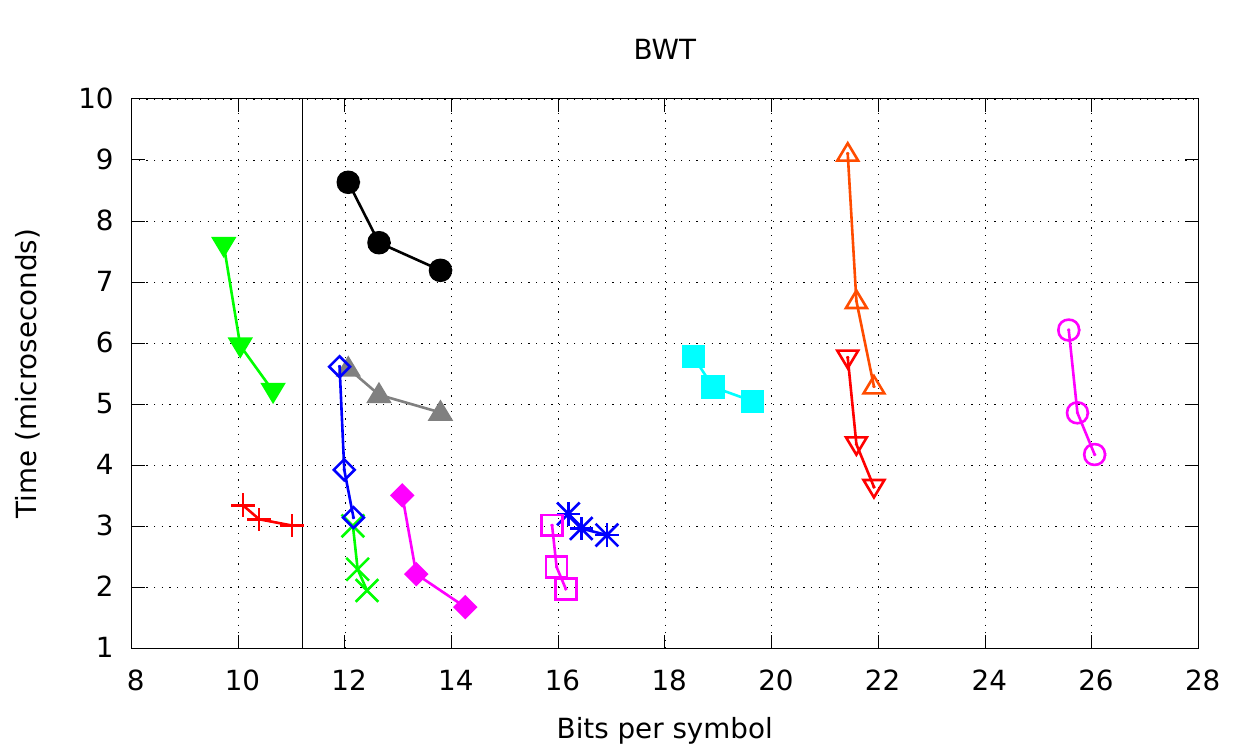}
\includegraphics[angle=90,scale=0.63]{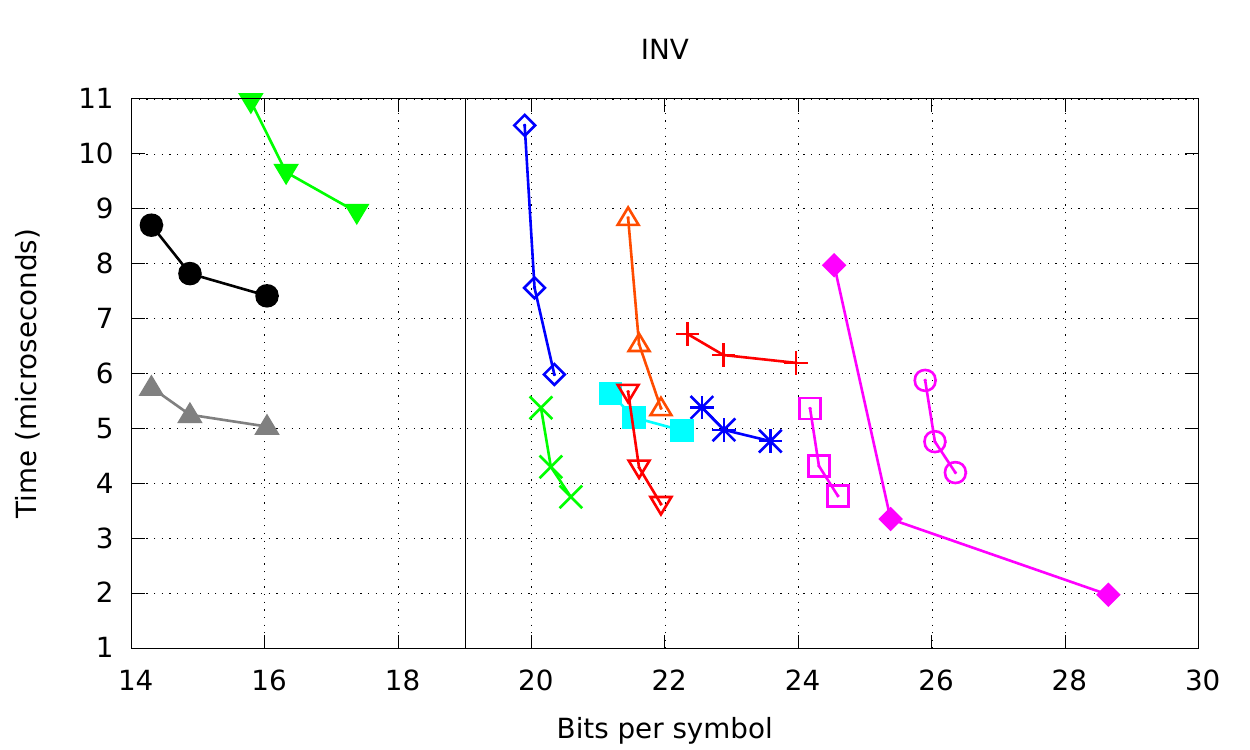} \\
\includegraphics[angle=90,scale=0.63]{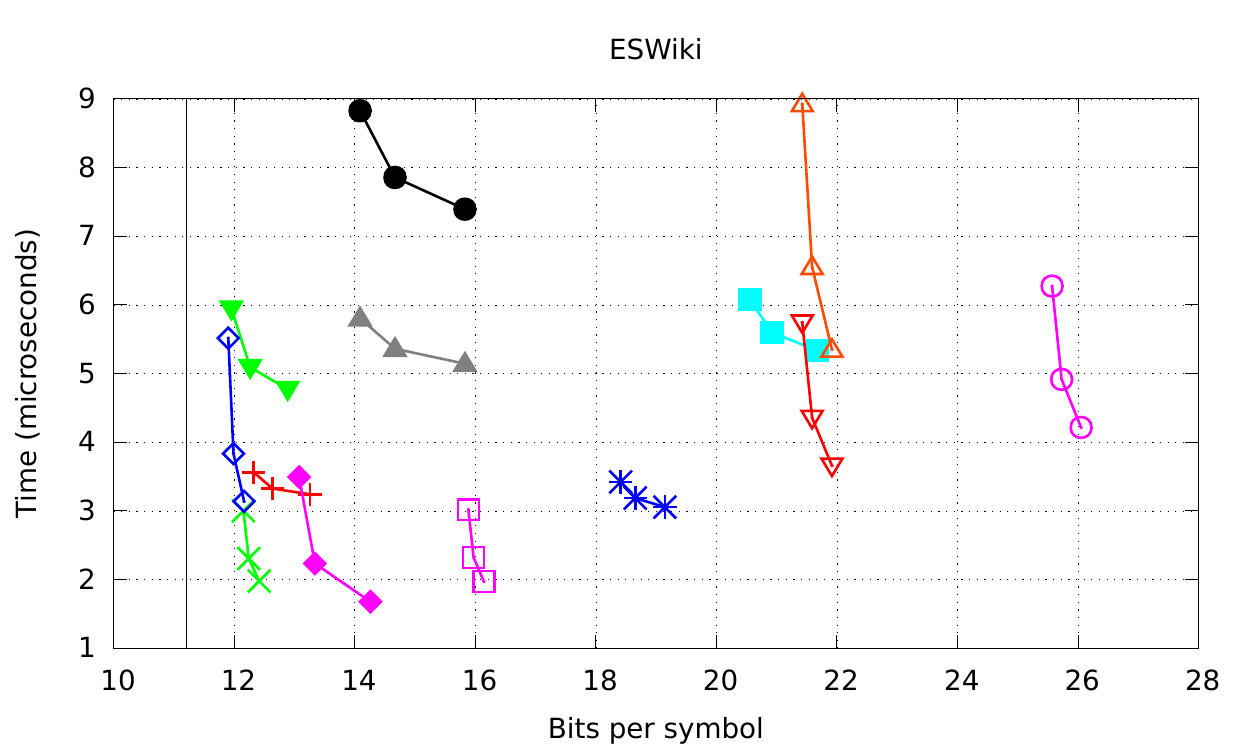}
\includegraphics[angle=90,scale=0.63]{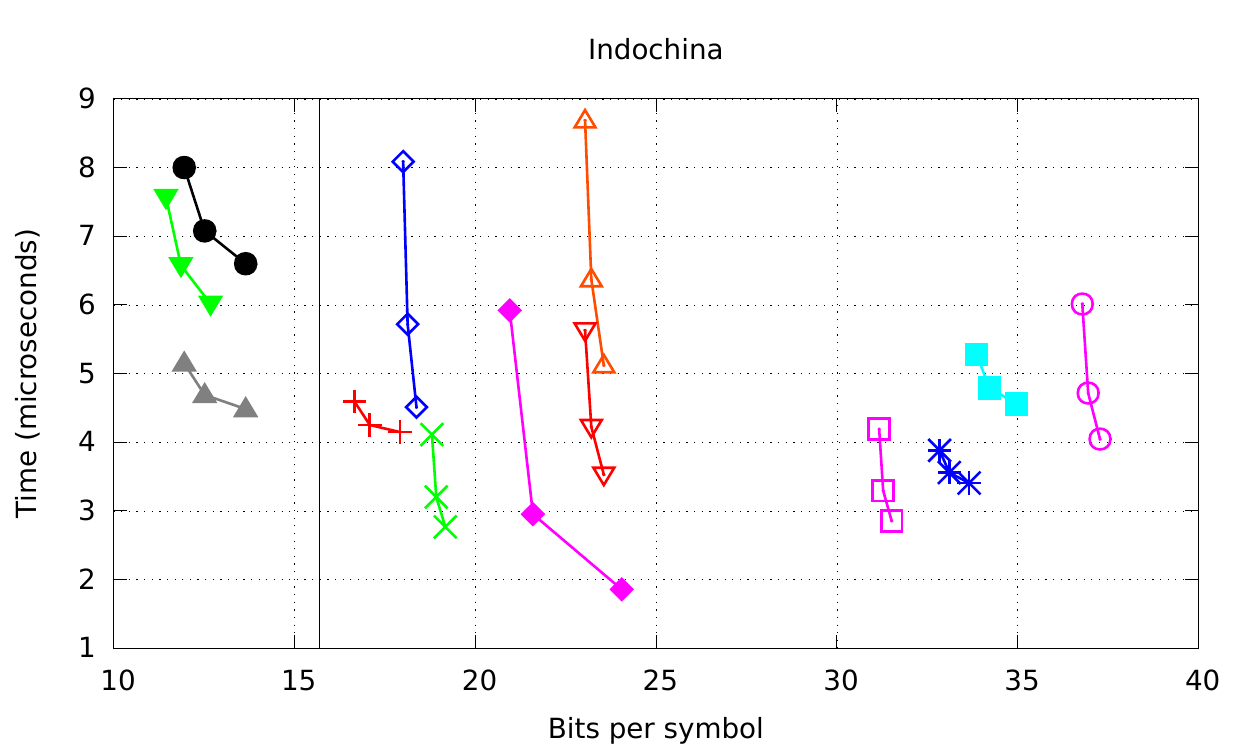}

\end{center}

\caption{Running time per $\access$ query over the four datasets.}
\label{chart:access}
\end{figure}

\begin{figure}[p]

    \begin{center}
    \includegraphics[angle=90,scale=1]{legend_sequences.pdf} \\    
    \includegraphics[angle=90,scale=0.63]{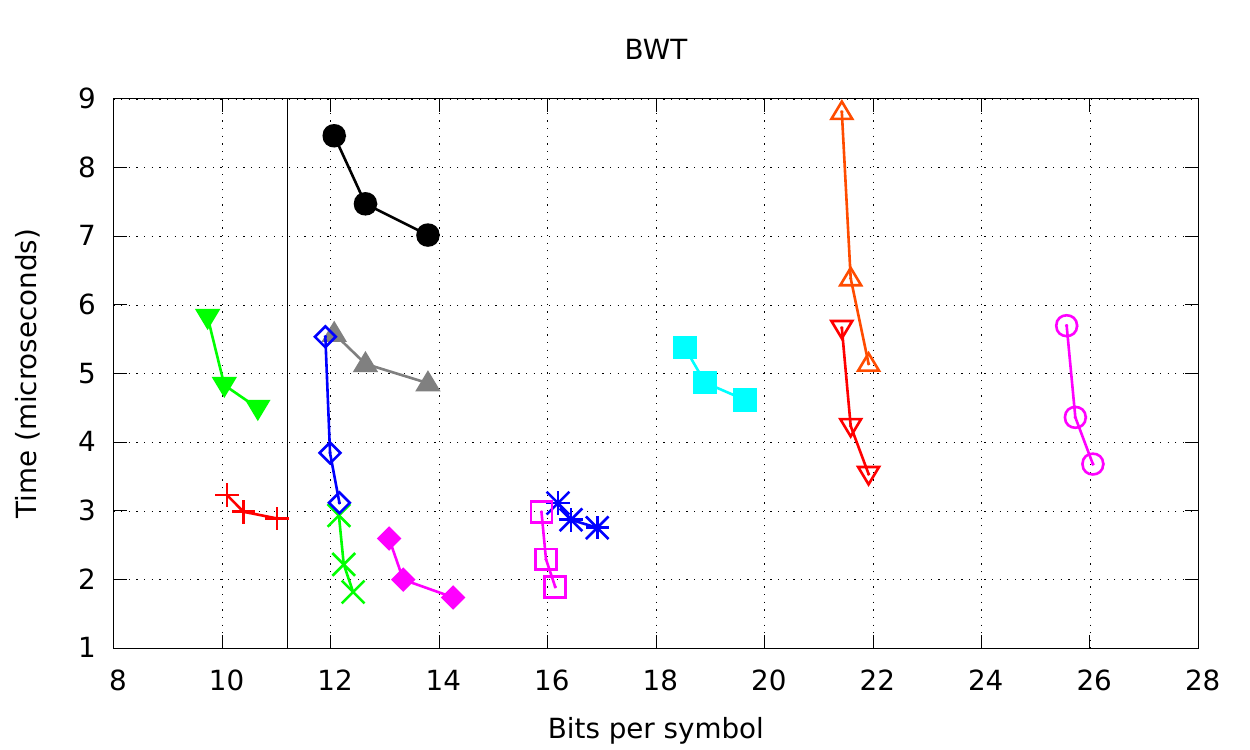}
    \includegraphics[angle=90,scale=0.63]{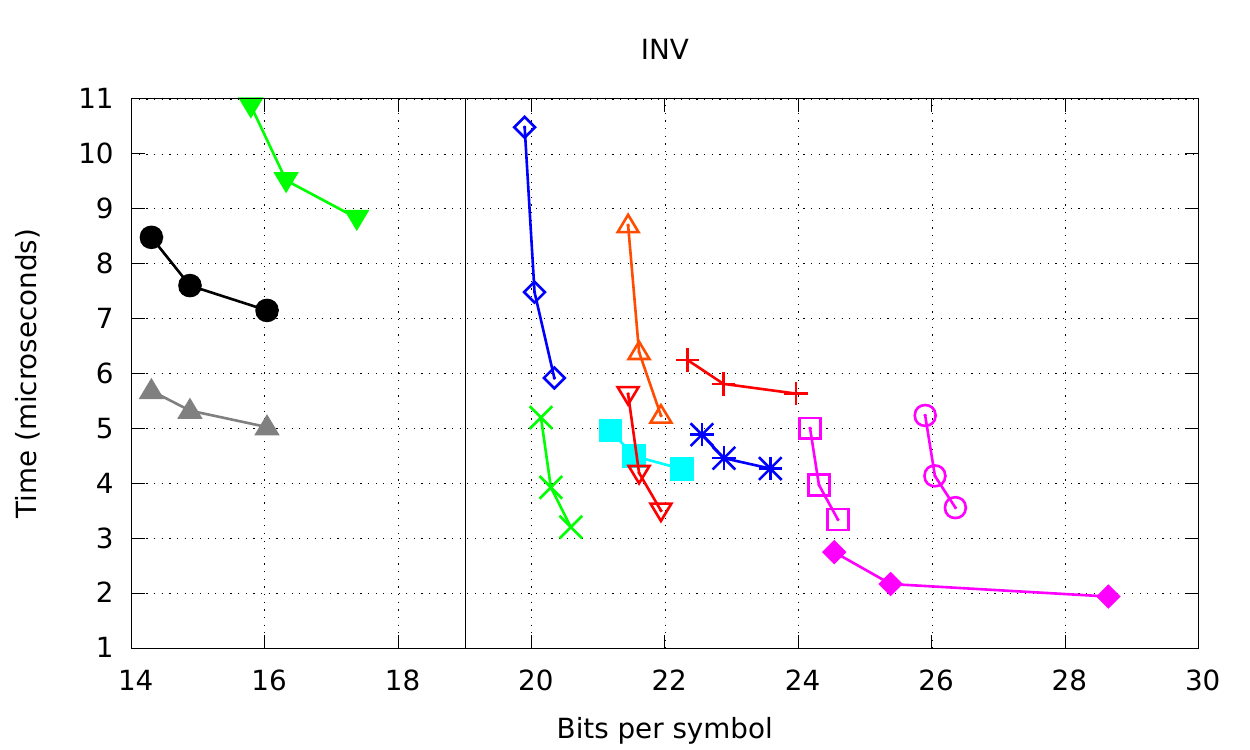} \\
    \includegraphics[angle=90,scale=0.63]{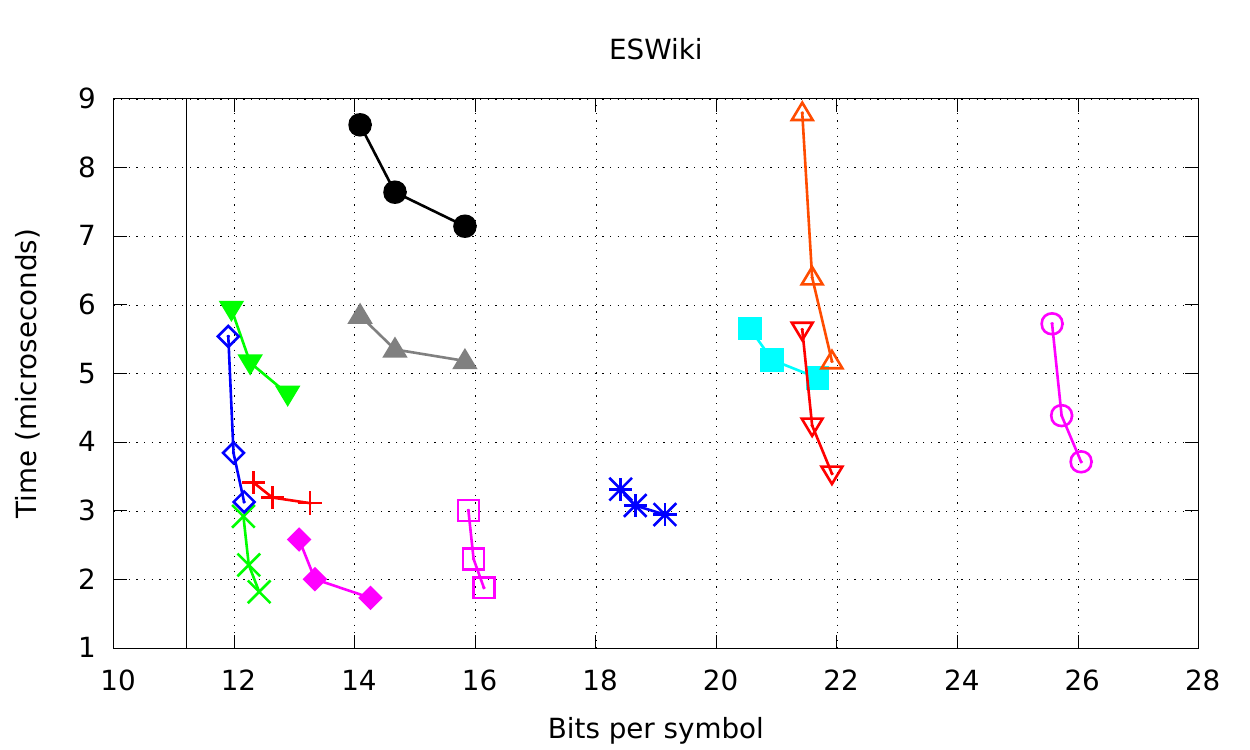}
    \includegraphics[angle=90,scale=0.63]{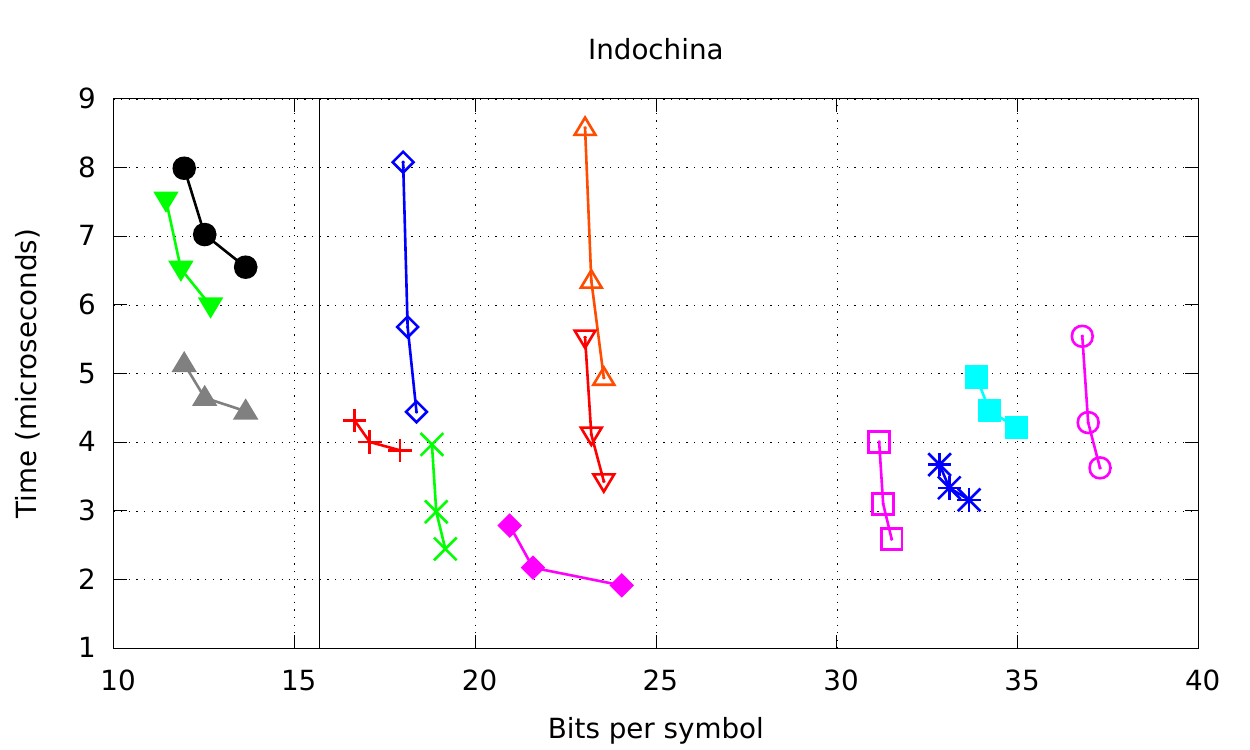}

    \end{center}

\caption{Running time per $\rank$ query over the four datasets.}
\label{chart:rank}
\end{figure}

\begin{figure}[p]
  \begin{center}
 \includegraphics[angle=90,scale=1]{legend_sequences.pdf} \\
  \includegraphics[angle=90,scale=0.63]{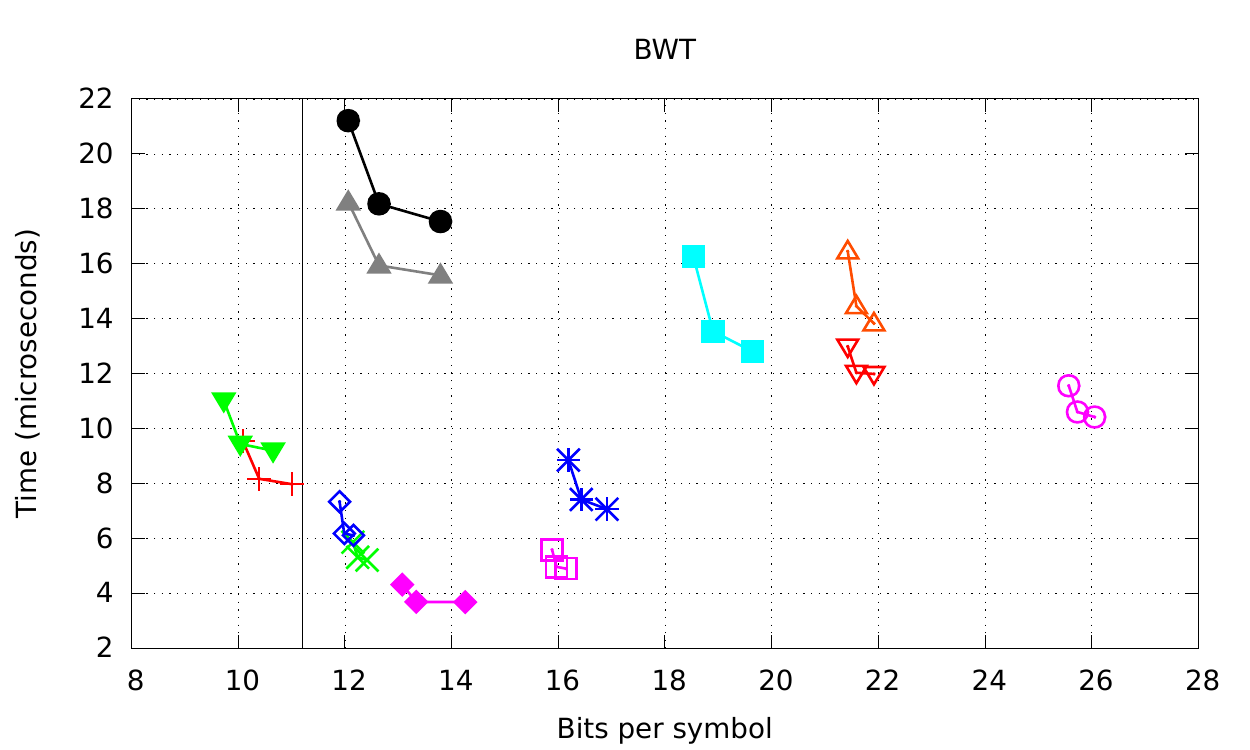}
  \includegraphics[angle=90,scale=0.63]{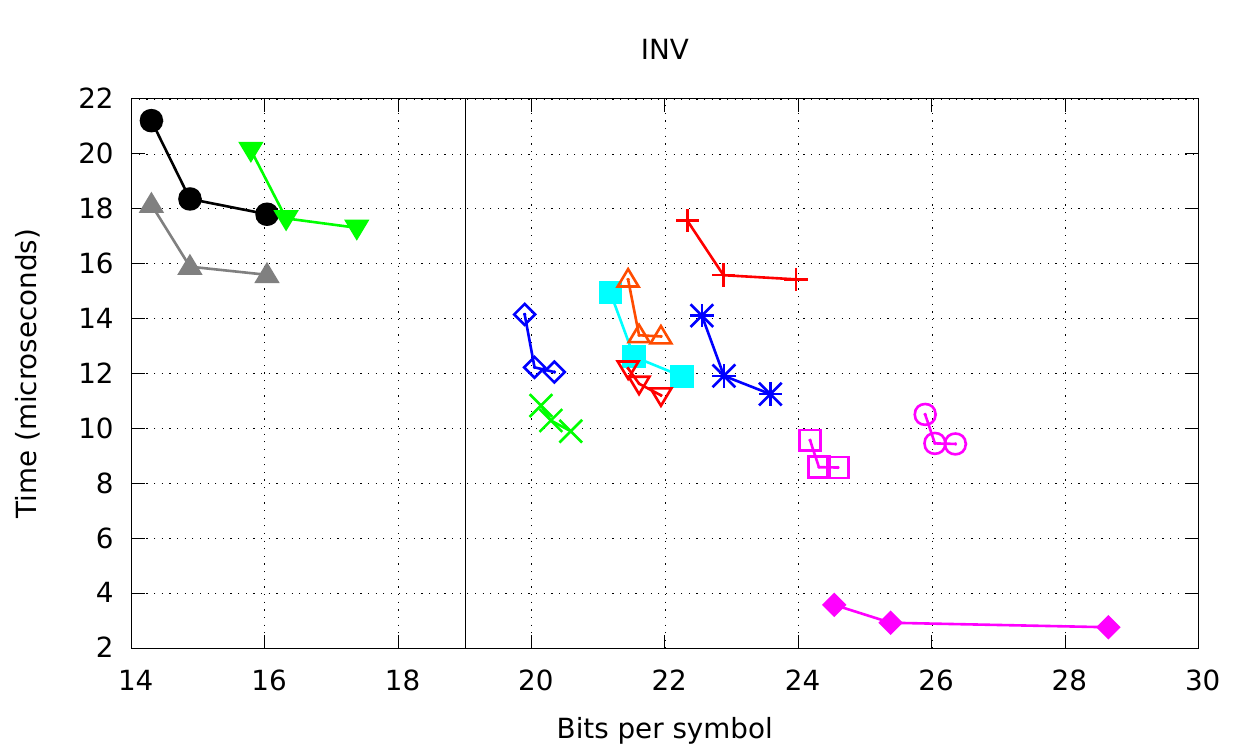} \\
  \includegraphics[angle=90,scale=0.63]{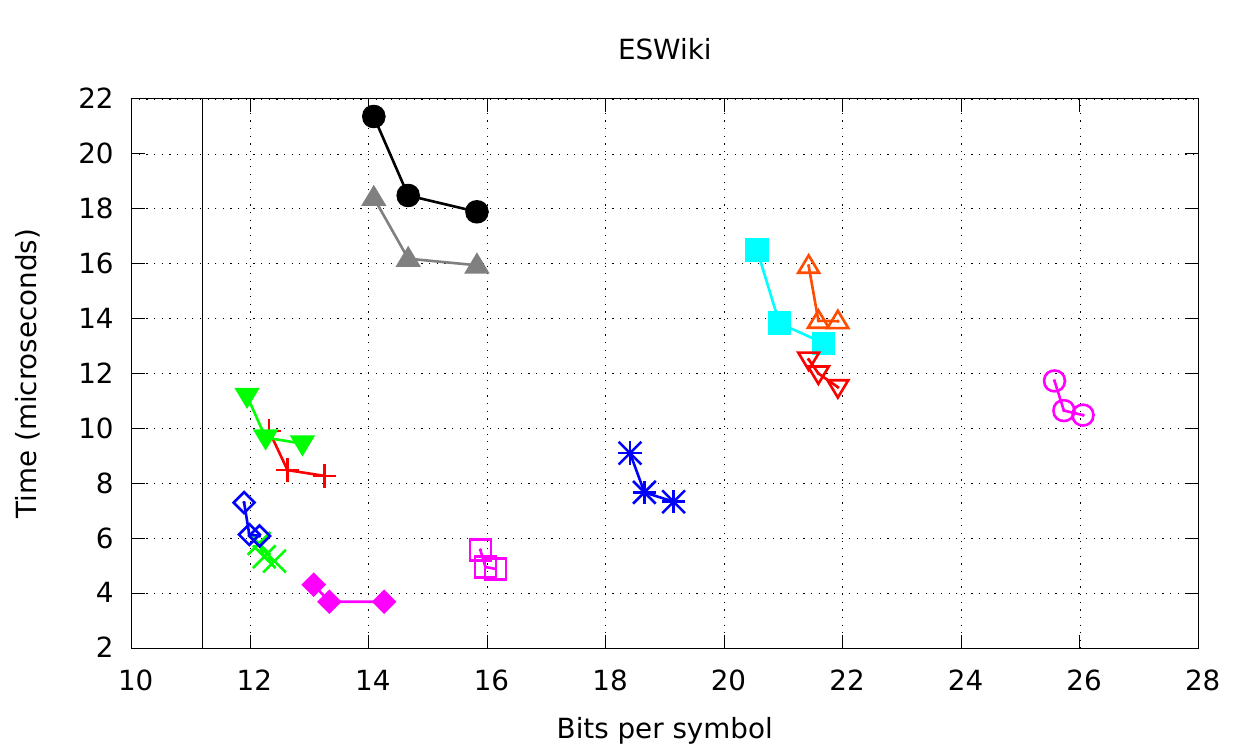}
  \includegraphics[angle=90,scale=0.63]{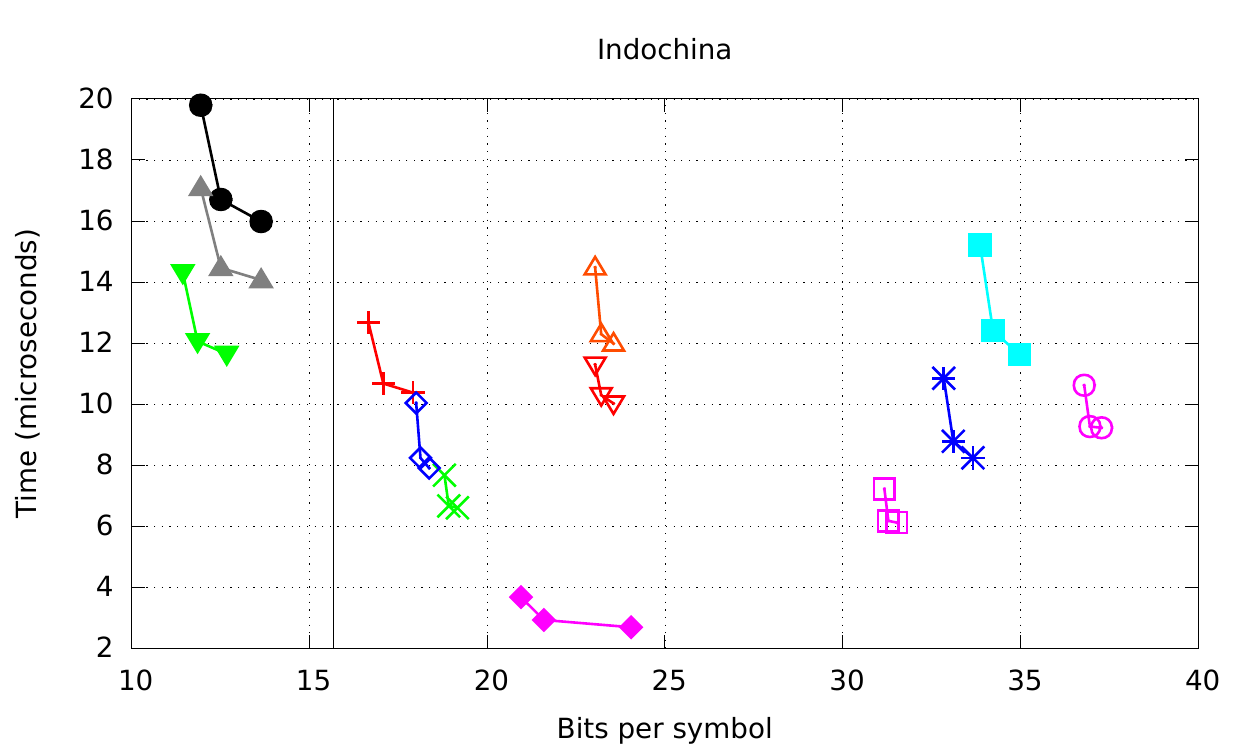}

  \end{center}

\caption{Running time per $\select$ query over the four datasets.}
\label{chart:select}
\end{figure}

\begin{figure}[p]
\begin{center}

\includegraphics[scale=0.8]{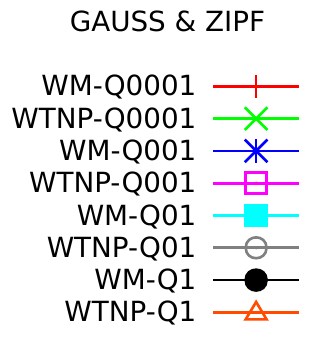} 
\hspace{1.8cm}
\includegraphics[scale=0.8]{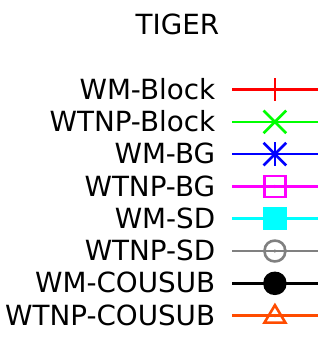}\\
\includegraphics[scale=0.7]{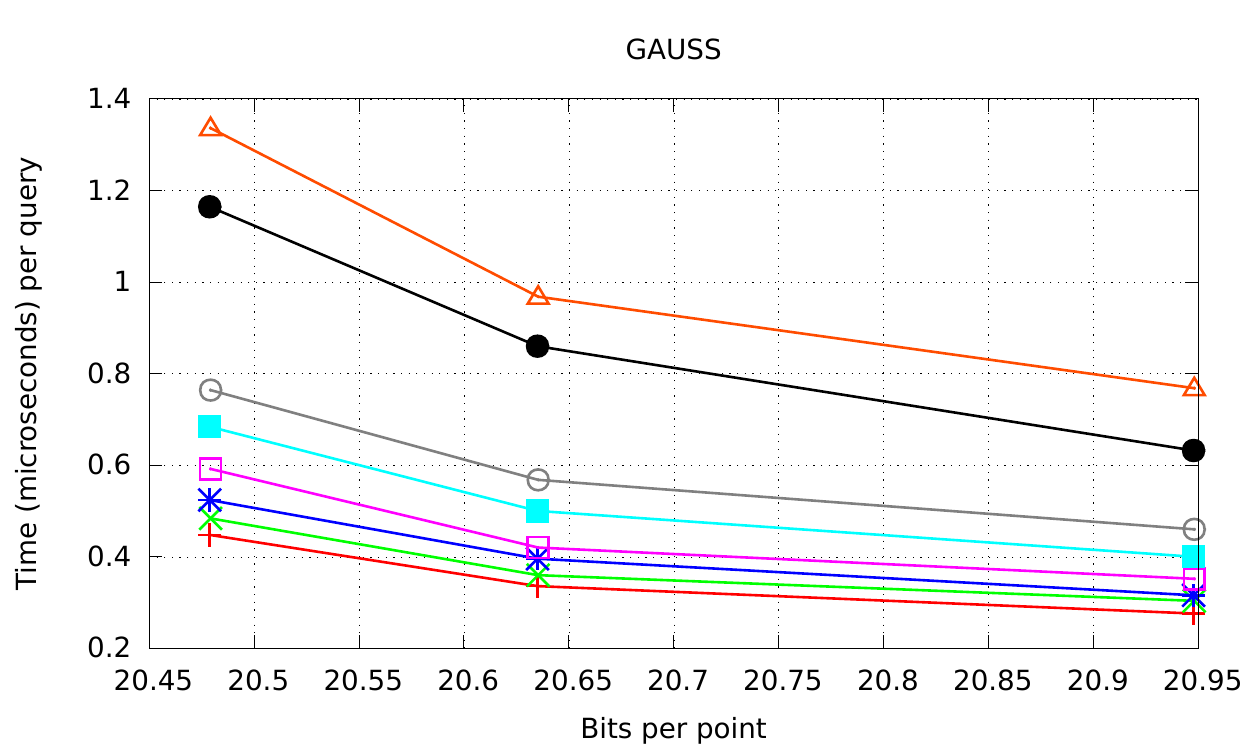}
\includegraphics[scale=0.7]{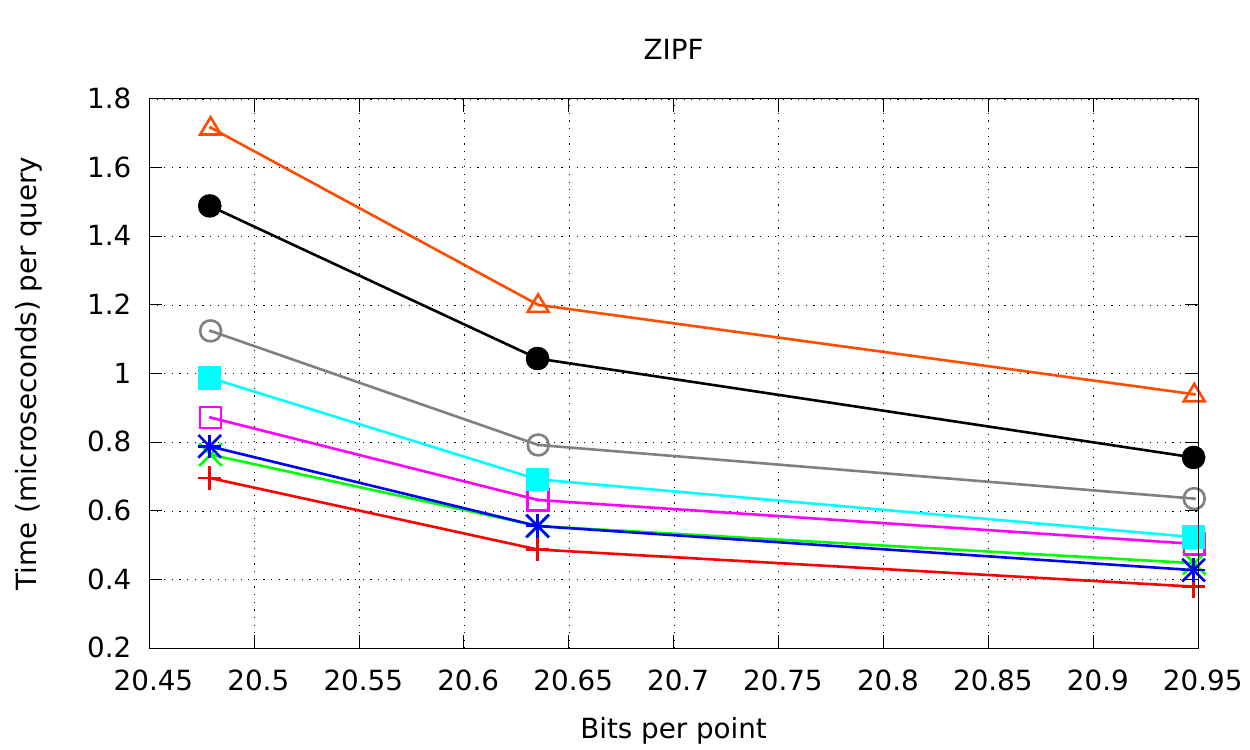}
\includegraphics[scale=0.7]{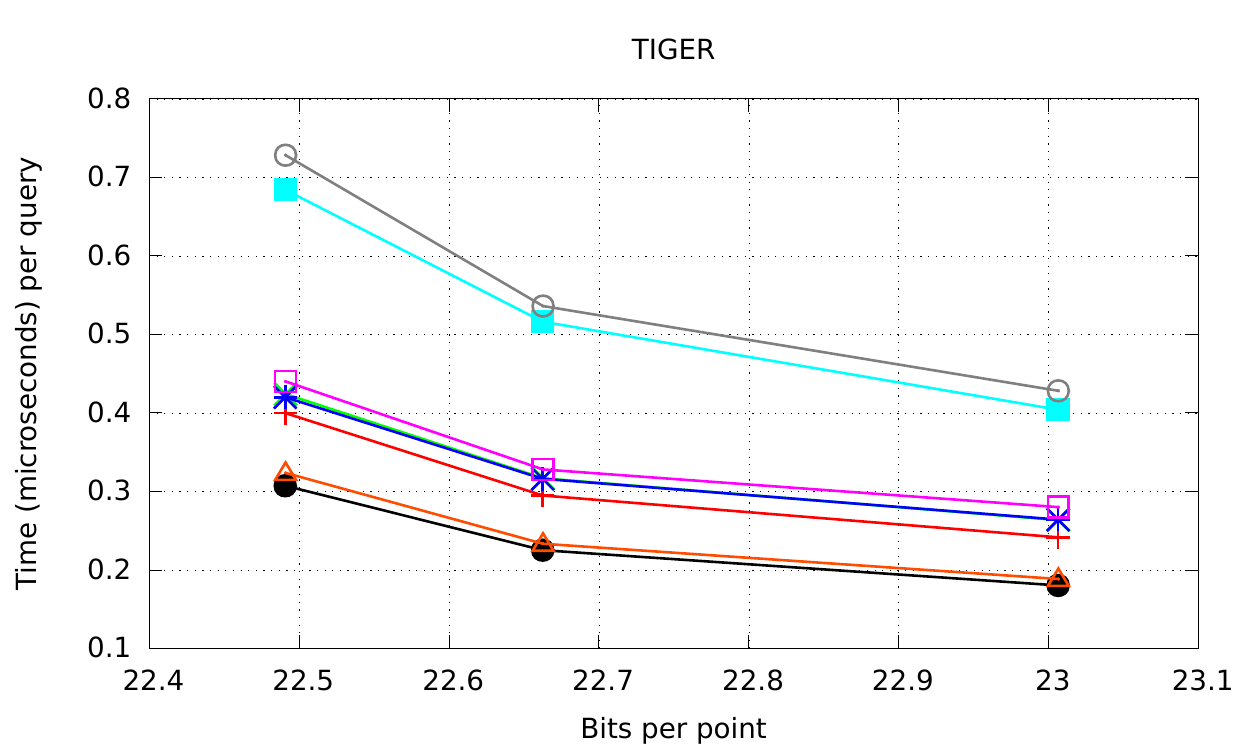}

\end{center}

\caption{Running time per $\rc$ query over the three datasets.}
\label{chart:range.count}
\end{figure}

\begin{figure}[p]
\begin{center}
\includegraphics[scale=0.8]{legend_gauss_zipf.pdf} 
\hspace{1.8cm}
\includegraphics[scale=0.8]{legend_tiger.pdf}\\
\includegraphics[scale=0.7]{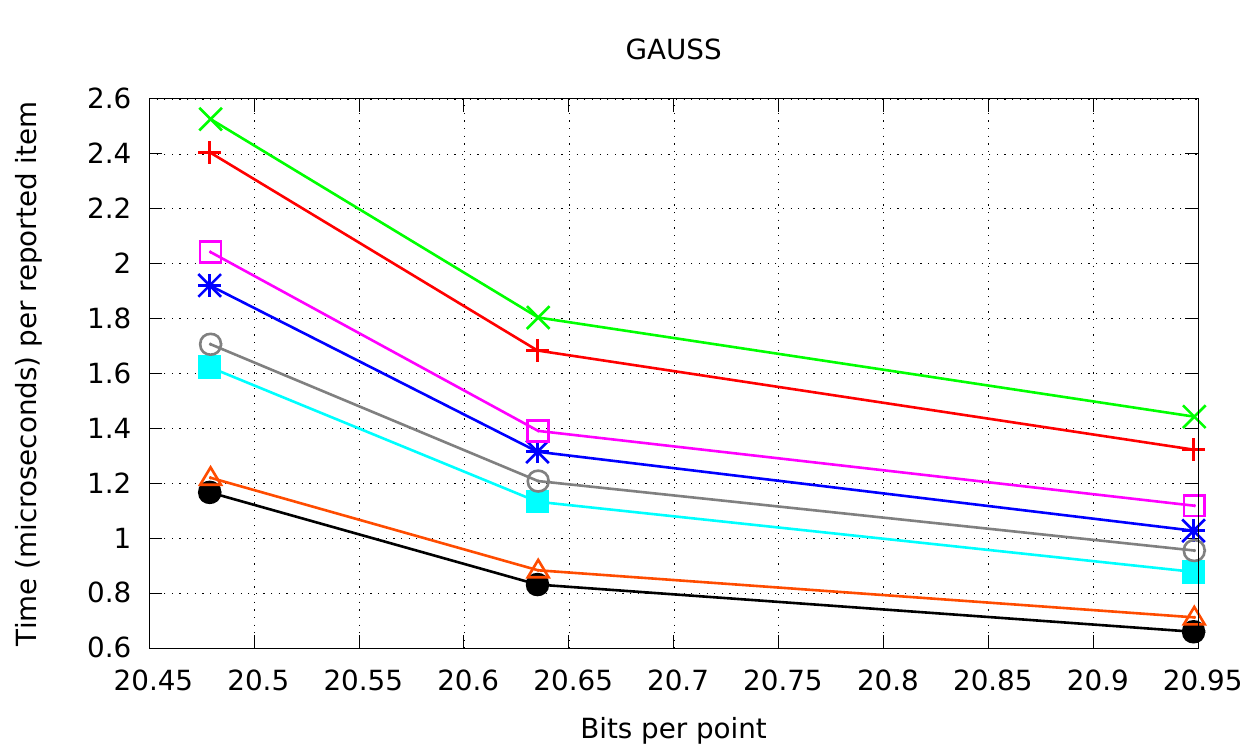}
\includegraphics[scale=0.7]{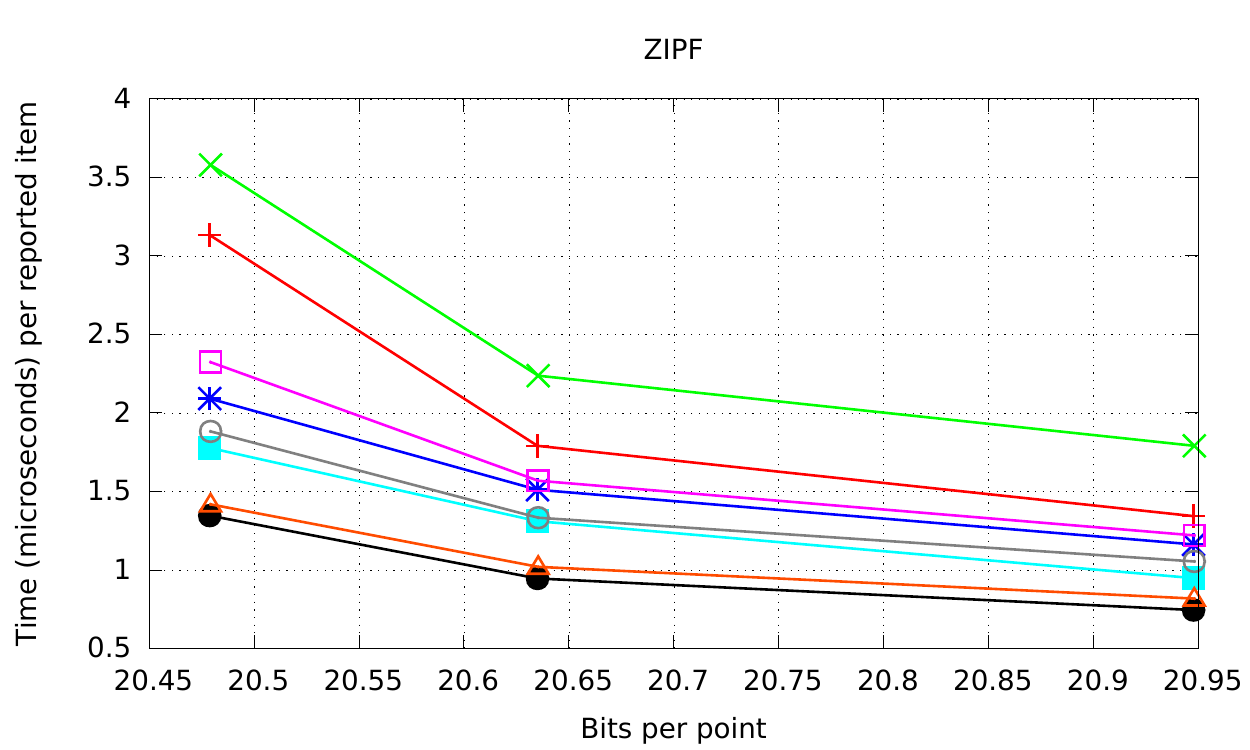}
\includegraphics[scale=0.7]{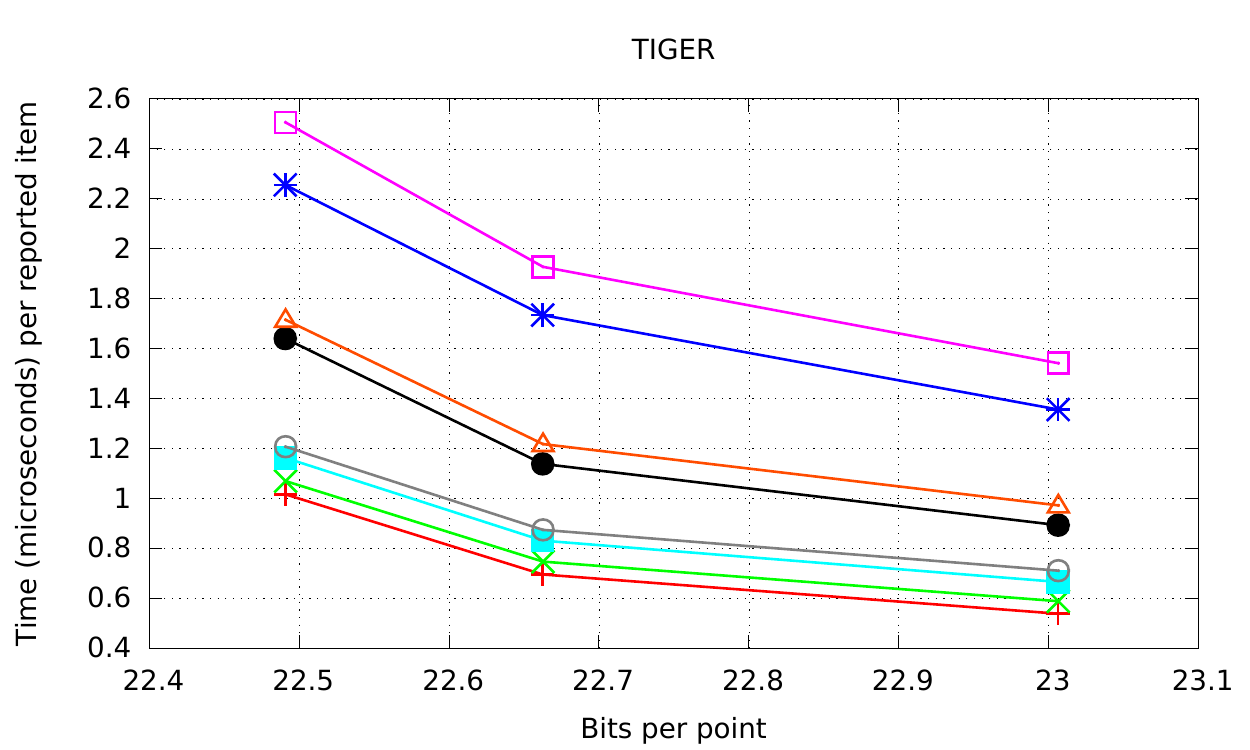}

\end{center}

\caption{Running time of $\rr$ query over the three datasets.}
\label{chart:range.report}
\end{figure}

The results for the counting queries shows that the time worsens as the query
are less selective.
The wavelet matrix is always faster than the pointerless wavelet tree, while
using the same space. The difference in time is proportional to the cost for
each selectivity, but additive with respect to the sampling. For example, it
becomes about 25\% faster when using the most space. We note in passing that
the space is basically 21 bps for the synthetic spaces and 23 bps for the
\verb|Tiger| dataset, which is essentially $\lg\sigma = \lg n$.

In the case of reporting queries, we show the time per reported item, which 
decreases as the query is less selective. Once again the wavelet matrix is
faster than the pointerless wavelet tree, albeit this time by a smaller
margin: usually below 10\%.

\section{Conclusions}

The levelwise wavelet tree \cite{MN06,MN07}, designed to avoid the
$O(\sigma\lg n)$ space overhead of standard wavelet trees \cite{GGV03}, was
unnecessarily slow in practice. We have redesigned this data structure so that
its time overhead over standard wavelet trees is significantly lower. The
result, dubbed {\em wavelet matrix}, enjoys all the good properties of 
levelwise wavelet trees but performs significantly faster in practice. It 
requires $n\lg\sigma + o(n\lg\sigma)$ bits of space, and can be built in 
$O(n\lg\sigma)$ time and almost in-place. We have also shown how to represent
Huffman shaped wavelet trees without using pointers by means of canonical
Huffman codes, and adapted the mechanism to Huffman shaped wavelet matrices.
This required a nontrivial redesign of the variable-length code assignment
mechanism. Our experimental results show that the compressed wavelet matrix
dominates the space/time tradeoff map for all the real-life sequences we
considered, also outperforming in most cases other structures designed for 
large alphabets
\cite{BCGNN13}. We also showed that the wavelet matrix is the best choice to
represent point grids that support orthogonal range queries.

An interesting future work is to adapt multiary wavelet trees \cite{FMMN07} to
wavelet matrices. The only difference is that, instead of a single accumulator 
$z_\ell$ per level, we have an array of $\rho-1$ accumulators in a $\rho$-ary 
wavelet matrix. As the useful values for $\rho$ are $O(\lg n)$, the overall 
space is still negligible, $O(\lg^2 n \lg \sigma)$. The real challenge is to
transalte the reduction in depth into a reduction of actual execution times.

Dynamic wavelet trees \cite{MN08,HM10,NS10} can immediately be translated into
wavelet matrices. It would be interesting to consider newer, theoretically
more efficient dynamic versions \cite{NN13}, and obtain practically efficient
implementations over wavelet matrices.

\paragraph{Acknowledgement.} 
Thanks to Daisuke Okanohara for useful comments.

\bibliographystyle{plain}
\bibliography{paper}

\end{document}